%% file: gaps.tex
\def\confversion{0}
\def\ifconf{\ifnum\confversion=1}
\def\ifnotconf{\ifnum\confversion=0}

\documentclass[11 pt]{article}

\def\showauthornotes{0}
\def\showkeys{0}
\def\showdraftbox{0}

\input{macros}

\newcommand{\mtnote}{\Authornote{MT}}

\usepackage[toc,page]{appendix}
\usepackage[capitalise,nameinlink]{cleveref}
\usepackage{thmtools}
\usepackage{thm-restate}

\Crefname{claim}{Claim}{Claims}

\newcommand{\sat}{{\sf sat}}

\renewcommand{\defeq}{:=}

\newcommand{\bnu}{\bar{\nu}}

\newcommand{\dist}{{\cal D}}
\newcommand{\dzero}{\dist^{(0)}}
\newcommand{\treedist}{\overline{\cal D}}
\newcommand{\component}[1]{\mathcal{C}\inparen{#1}}

\newcommand{\Hyp}{\calH_k\inparen{m,n,n_0,\Gamma}}
\newcommand{\bigoh}{\operatorname{O}}

\newcommand{\dmu}{\ensuremath{d_\mu}}

\newcommand{\Vtx}{\operatorname{V}}
\newcommand{\Edg}{\operatorname{E}}

\newcommand{\cl}{\operatorname{cl}}
\newcommand{\clR}{\cl_R}

\newcommand{\ttg}{\mathtt{g}}
\newcommand{\girth}{\ttg}
\newcommand{\sgn}{\operatorname{sign}}
\newcommand{\littleoh}{\operatorname{o}}

\title{From Weak to Strong LP Gaps for all CSPs}
\author{
Mrinalkanti Ghosh\thanks{%
Toyota Technological Institute at Chicago. 
{\tt mkghosh@ttic.edu}
}
\and
Madhur Tulsiani\thanks{%
    Toyota Technological Institute at Chicago.  {\tt madhurt@ttic.edu.}
  }%
}

\begin{document}
\sloppy
\maketitle
\draftbox
\begin{abstract}
We study the approximability of constraint satisfaction problems (CSPs) by linear programming (LP) 
relaxations. 
We show that for every CSP, the approximation obtained by a basic LP
relaxation, is no weaker than the approximation obtained using relaxations given by 
$\Omega\inparen{\frac{\log n}{\log \log n}}$ levels of the Sherali-Adams hierarchy on instances of
size $n$. 

\smallskip
It was proved by Chan \etal [FOCS 2013] that any polynomial size LP extended formulation is no stronger
than relaxations obtained by a super-constant levels of the Sherali-Adams hierarchy.
Combining this with our result also implies that any polynomial size LP extended
formulation is no stronger than the basic LP.

\smallskip
Using our techniques, we also simplify and strengthen the result by Khot \etal [STOC 2014] on
(strong) approximation resistance for LPs. They provided a necessary and sufficient condition under
which $\Omega(\log \log n)$ levels of the Sherali-Adams hierarchy cannot achieve an approximation
better than a random assignment. We simplify their proof and strengthen the bound to
$\Omega\inparen{\frac{\log n}{\log \log n}}$ levels.
\end{abstract}
%


\section{Introduction}\label{sec:intro}
\input{intro.tex}

\section{Preliminaries}\label{sec:prelims}
\input{prelims.tex}

\section{Properties of random hypergraphs}\label{sec:hypergraphs}
\input{hypergraphs.tex}

\section{Decompositions of hypergraphs from local geometry}\label{sec:decompositions}
\input{decompositions.tex}

\section{The Sherali-Adams Integrality Gaps construction}\label{sec:sagaps}
\input{sagaps.tex}

\section*{Acknowledgements}
We thank Chandra Chekuri, Subhash Khot and Yury Makarychev for helpful discussions, and Rishi Saket
for pointers to references.
This research was supported by supported by the National Science Foundation under award number
CCF-1254044.

\bibliographystyle{plain}
\bibliography{macros,madhur}
\appendix
\input{appendix.tex}

%
\end{document}

%% file: macros.tex
\usepackage{xspace,xcolor,enumerate}
\usepackage{amsmath,amssymb}
\usepackage{color,graphicx}

\ifnum\showkeys=1
\usepackage[color]{showkeys}
\fi

\definecolor{darkred}{rgb}{0.5,0,0}
\definecolor{darkgreen}{rgb}{0,0.5,0}
\definecolor{darkblue}{rgb}{0,0,0.5}

\ifnotconf
\usepackage[pdfstartview=FitH,pdfpagemode=None,colorlinks,linkcolor=darkred,filecolor=blue,citecolor=darkred,urlcolor=darkred,pagebackref]{hyperref}
\fi

\usepackage[T1]{fontenc}
\usepackage{mathtools,dsfont,bbm}
\usepackage{mathpazo}

\ifnum\showauthornotes=1
\newcommand{\Authornote}[2]{{\sf\small\color{red}{[#1: #2]}}}
\newcommand{\Authorcomment}[2]{{\sf \small\color{gray}{[#1: #2]}}}
\newcommand{\Authorfnote}[2]{\footnote{\color{red}{#1: #2}}}
\else
\newcommand{\Authornote}[2]{}
\newcommand{\Authorcomment}[2]{}
\newcommand{\Authorfnote}[2]{}
\fi

\ifnum\showdraftbox=1
\newcommand{\draftbox}{\begin{center}
  \fbox{%
    \begin{minipage}{2in}%
      \begin{center}%
        \begin{Large}%
          \textsc{Working Draft}%
        \end{Large}\\
        Please do not distribute%
      \end{center}%
    \end{minipage}%
  }%
\end{center}
\vspace{0.2cm}}
\else
\newcommand{\draftbox}{}
\fi

\newtheorem{theorem}{Theorem}[section]

\newtheorem{definition}[theorem]{Definition}
\newtheorem{lemma}[theorem]{Lemma}

\newtheorem{corollary}[theorem]{Corollary}
\newtheorem{claim}[theorem]{Claim}

\def\FullBox{\hbox{\vrule width 6pt height 6pt depth 0pt}}

\def\qed{\ifmmode\qquad\FullBox\else{\unskip\nobreak\hfil
\penalty50\hskip1em\null\nobreak\hfil\FullBox
\parfillskip=0pt\finalhyphendemerits=0\endgraf}\fi}

\def\qedsketch{\ifmmode\Box\else{\unskip\nobreak\hfil
\penalty50\hskip1em\null\nobreak\hfil$\Box$
\parfillskip=0pt\finalhyphendemerits=0\endgraf}\fi}

\ifnotconf
\newenvironment{proof}{\begin{trivlist} \item {\bf Proof:~~}}
   {\qed\end{trivlist}}
\fi

\newenvironment{proofof}[1]{\begin{trivlist} \item {\bf Proof
#1:~~}}
  {\qed\end{trivlist}}

\def\to{\rightarrow}
\def\eps{\varepsilon}
\def\epsilon{\varepsilon}

\def\eps{\epsilon}

\def\phi{\varphi}
\def\cal{\mathcal}

\def\implies{\Rightarrow}

\newcommand{\defeq}{\stackrel{\mathrm{def}}=}

\renewcommand{\bar}{\overline} 

\newcommand{\ie}{i.e.,\xspace}
\newcommand{\eg}{e.g.,\xspace}
\newcommand{\etal}{et al.\xspace}

\newcommand{\mper}{\,.}
\newcommand{\mcom}{\,,}

\newcommand{\R}{{\mathbb R}}

\newcommand{\N}{{\mathbb{N}}}

\newcommand{\Q}{{\mathbb{Q}}}

\newcommand{\B}{\{0,1\}\xspace}

\newcommand{\indicator}[1]{\mathds{1}_{\{#1\}}}

\usepackage{nicefrac}

\let\nfrac=\nicefrac

\newcommand{\abs}[1]{\ensuremath{\left\lvert #1 \right\rvert}}

\newcommand{\norm}[1]{\ensuremath{\left\lVert #1 \right\rVert}}

\def\bfu{{\bf u}}

\newcommand{\qary}{[q]}

\newcommand{\vartwo}[2]{x_{(#1, #2)}}
\newcommand{\vartwoempty}{x_{(\emptyset,\emptyset)}}

\newcommand{\Esymb}{\mathbb{E}}
\newcommand{\Psymb}{\mathbb{P}}

\DeclareMathOperator*{\ExpOp}{\Esymb}

\DeclareMathOperator*{\ProbOp}{\Psymb}
\renewcommand{\Pr}{\ProbOp}

\newcommand{\prob}[1]{\Pr\left[{#1}\right]}
\newcommand{\Prob}[2]{\Pr_{{#1}}\left[{#2}\right]}
\newcommand{\ex}[1]{\ExpOp\left[{#1}\right]}
\newcommand{\Ex}[2]{\ExpOp_{{#1}}\left[{#2}\right]}

\def\expop{\ExpOp}

\newfont{\inhead}{eufm10 scaled\magstep1}
\newcommand{\deffont}{\sf}

\newcommand{\calC}{{\cal C}}

\newcommand{\calE}{{\cal E}}

\newcommand{\calH}{{\cal H}}

\newcommand{\calP}{{\cal P}}

\newcommand{\suchthat}{{\;\; : \;\;}}

\DeclareMathOperator\supp{Supp}

\newcommand{\opt}{{\sf OPT}\xspace}
\newcommand{\sdpopt}{{\sf FRAC}\xspace}
\newcommand{\lpopt}{{\sf FRAC}\xspace}

\newcommand{\problemmacro}[1]{\textsf{#1}}

\newcommand{\maxcsp}{\problemmacro{MAX-CSP}\xspace}
\newcommand{\maxkcsp}{\problemmacro{MAX k-CSP}\xspace}
\newcommand{\maxkcspq}{\problemmacro{MAX k-CSP}$_q$\xspace}

\newcommand{\maxthreesat}{\problemmacro{MAX 3-SAT}\xspace}
\newcommand{\maxthreexor}{\problemmacro{MAX 3-XOR}\xspace}

\newcommand{\maxcut}{\problemmacro{MAX-CUT}\xspace}

\newcommand{\vertexcover}{\problemmacro{Minimum Vertex Cover}\xspace}

\newcommand{\inparen}[1]{\left(#1\right)}             
\newcommand{\inbraces}[1]{\left\{#1\right\}}           
\newcommand{\insquare}[1]{\left[#1\right]}             

\newcommand{\Lovasz}{Lov\'asz\xspace}

\newcommand{\Bollobas}{Bollob\'as\xspace}

%% file: intro.tex
Given a finite alphabet $\qary = \inbraces{0, \ldots, q-1}$ and a predicate $f: \qary^k \to \B$, an
instance of the problem $\maxkcsp(f)$ consists of (say) $m$ constraints over a set of $n$ variables
$x_1, \ldots, x_n$ taking values in the set $[q]$. 
%
Each constraint $C_i$ is of the form
$f(x_{i_1}+b_{i_1}, \ldots, x_{i_k} + b_{i_k})$ for some $k$-tuple of variables $(x_{i_1}, \ldots
x_{i_k})$ and $b_{i_1}, \ldots, b_{i_q}  \in \qary$, and the addition is taken to be modulo
$q$. We say an assignment $\sigma$ to the variables satisfying the constraint $C_i$ if
$C_i(\sigma(x_{i_1}), \ldots, \sigma(x_{i_k})) = 1$. Given an instance $\Phi$ of the problem, the
goal is to find an assignment $\sigma$ to the variables satisfying as many constraints as possible.
The approximability of the $\maxkcsp(f)$ problem has been extensively studied for various predicates $f$
(see \eg \cite{Hastad07:survey} for a survey), and special cases include several interesting and
natural problems such as \maxthreesat, \maxthreexor and \maxcut.

A topic of much recent interest has been the efficacy of Linear Programming (LP) and Semidefinite
Programming (SDP) relaxations. For a given instance $\Phi$ of $\maxkcsp(f)$, let $\opt(\Phi)$ denote
the \emph{fraction} of constraints satisfied by an optimal assignment, and let $\lpopt(\Phi)$ denote
the value of the convex (LP/SDP) relaxation for the problem. Then, the performance guarantee of this
algorithm is given by the {\deffont integrality gap} which equals the supremum of
$\frac{\lpopt(\Phi)}{\opt(\Phi)}$, over all instances $\Phi$. 

The study of unconditional lower bounds for general families of LP relaxations was initiated by
Arora, \Bollobas and \Lovasz \cite{AroraBL02} (see also \cite{AroraBLT06}). They studied the
\Lovasz-Schrijver \cite{LoS91} LP hierarchy and proved lower bounds on the integrality gap for
\vertexcover (their technique also yields similar bounds for \maxcut).
De la Vega and Kenyon-Mathieu
\cite{delaVegaK07} and Charikar, Makarychev and Makarychev \cite{CharikarMM09} proved a lower bound
of $2-o(1)$ for the integrality gap of the LP relaxations for \maxcut given respectively by
$\Omega(\log \log n)$ and $n^{\Omega(1)}$ levels  of the Sherali-Adams LP hierarchy \cite{SA90}. 
Several follow-up works have also shown lower bounds for various other special cases of the \maxkcsp
problem, both for LP and SDP hierarchies \cite{AroraAT05, Schoenebeck08, Tulsiani09, RaghavendraS09,
  BenabbasGMT12, BarakCK15}.

A recent result by Chan \etal \cite{ChanLRS13} shows a connection between strong lower bounds for the
Sherali-Adams hierarchy, and lower bounds on the size of LP extended formulations for the
corresponding problem. 
%
%
In fact, their result proved a connection
not only for a lower bound on the worst case integrality gap, but for the entire
\emph{approximability curve}. We say that $\Phi$ is {\deffont $(c,s)$-integrality gap instance} for a
relaxation of $\maxkcsp(f)$, if we have $\lpopt(\Phi) \geq c$ and $\opt(\Phi) < s$. They showed
that for any fixed $t \in \N$, if there exist $(c,s)$-integrality gap instances of size $n$ for the
relaxation given by $t$ levels of the Sherali-Adams hierarchy, then for all $\eps > 0$ and
sufficiently large $N$, there exists a $(c -\eps,s+\eps)$ integrality gap instance of size (number
of variables) $N$, 
for any linear extended formulation of size at most $N^{t/2}$. They also give a tradeoff (described
later) when $t$ is a function of $n$, which was recently improved by Kothari \etal \cite{KothariMR16}.

We strengthen the above results by showing that for all $c,s \in [0,1]$,  
$(c,s)$-integrality gap instances for a ``basic LP'' can be used to construct $(c-\eps,s+\eps)$ integrality
gap instances for $\Omega_{\eps}\inparen{\frac{\log n}{\log \log n}}$ levels of the Sherali-Adams
hierarchy. The basic LP uses only a subset of the constraints in the relaxation given by $k$ levels
of the Sherali-Adams hierarchy for $\maxkcsp(f)$. In particular, this shows that a lower bound on
the integrality gap for the basic LP, implies a similar lower bound on the integrality gap of any
polynomial size extended formulation. We note that both the above results 
and our result apply to all $f, q$ and all $c,s \in [0,1]$.

\paragraph{Comparison with (implications of) Raghavendra's UGC hardness result.} 
A remarkable result by Raghavendra \cite{Raghavendra08} shows that a $(c,s)$-integrality gap
instance for a ``basic SDP'' relaxation of $\maxkcsp(f)$ implies hardness of distinguishing
instances $\Phi$ with $\opt(\Phi) < s$ from instances with $\opt(\Phi) \geq c$, assuming the Unique
Games Conjecture (UGC) of Khot \cite{Khot02:unique}. The basic SDP considered by
Raghavendra involves moments for all pairs of variables, and all subsets of variables included in a
constraint. 
The basic LP we consider is weaker than this SDP and does not contain the positive
semidefiniteness constraint.

Combining Raghavendra's result with known
constructions of integrality gaps for Unique Games by Raghavendra and Steurer \cite{RaghavendraS09},
and by Khot and Saket \cite{KhotS09}, one can obtain a result qualitatively similar to ours, for the
mixed hierarchy. In particular, a $(c,s)$ integrality gap for the basic SDP implies a
$(c-\eps,s+\eps)$ integrality gap for $\Omega((\log \log n)^{1/4})$ levels of the mixed hierarchy. 

Note however, that the above result is incomparable to our result, since it starts with stronger
hypothesis (a basic SDP gap) and yields a gap for the mixed hierarchy as opposed to the
Sherali-Adams hierarchy. While the above can also be used to derive lower bounds for linear extended
formulations, one needs to start with an SDP gap instance to derive an LP lower bound. The basic SDP
is known to be provably stronger than the basic LP for several problems including various
2-CSPs. Also, for the worst case $f$ for $q=2$, the integrality gap of the basic SDP is $O(2^k/k)$
\cite{CMM07:csp}, while that of the basic LP is $2^{k-1}$.

A recent result by Khot and Saket \cite{KhotS15} shows a connection between the integrality
gaps for the basic LP and those for the basic SDP. 
They prove that, assuming the UGC, a $(c,s)$ integrality gap instance for the basic LP implies an
NP-hardness of distinguishing instances $\Phi$ with $\opt(\Phi) \geq \Omega\inparen{\frac{c}{k^3
    \cdot \log(q)}}$ from instances with $\opt(\Phi) \leq 4s$.
Their result also shows that a $(c,s)$ integrality gap instance for the basic LP can be used to
produce  a $\inparen{\Omega\inparen{\frac{c}{k^3 \cdot \log(q)}}, 4s}$ integrality gap instance for
the basic SDP,  and hence for $\Omega((\log \log n)^{1/4})$ levels of the mixed hierarchy.

\paragraph{Other related work.}
The power of the basic LP for solving valued CSPs \emph{to optimality} has been studied in several
previous works. These works consider the problem of minimizing the penalty for unsatisfied
constraints, where the penalties take values in $\Q \cup \inbraces{\infty}$. Also, they study the
problem not only in terms of single predicate $f$, but rather in terms of the constraint
language generated by a given set of (valued) predicates. 

It was shown by Thapper and \v{Z}ivn\'{y} \cite{ThapperZ13} that when the penalties are
finite-valued, if the problem of finiding the optimum solution cannot be solved by the basic LP,
then it is NP-hard. Kolmogorov, Thapper and \v{Z}ivn\'{y} \cite{KolmogorovTZ15} give a
characterization of CSPs where the problem of minimizing the penalty for unsatisfied constraints can
be solved \emph{exactly} by the basic LP. Also, a recent result by Thapper and \v{Z}ivn\'{y}
\cite{ThapperZ16} shows the valued CSP problem for a constraint language can be solved to optimality
by a bounded number of levels of the Sherali-Adams hierarchy if and only if it can be solved by a
relaxation obtained by augmenting the basic LP with contraints implied by three levels of the
Sherali-Adams hierarchy. However, the above works only consider the case when the LP gives an exact
solution, and do not focus on approximation.

The techniques from \cite{CharikarMM09} used in our result, were also extended by Lee \cite{Lee15}
to prove a hardness for the Graph Pricing problem. Kenkre \etal \cite{KenkrePPS15} also applied these to
show the optimality of a simple LP-based algorithm for Digraph Ordering.

\subsubsection*{Our results}
Our main result is the following.
%
%
\begin{restatable}{reptheorem}{maintheorem}
\label{thm:main}
Let $f: \qary^k \to \B$ be any predicate. Let $\Phi_0$ be a $(c,s)$ integrality gap instance for basic
LP relaxation of \maxkcsp($f$). Then for every $\eps > 0$, there exists $c_{\eps} > 0$ 
such that for infinitely many $N \in \N$, there exist $(c-\eps, s+\eps)$ integrality gap
instances of size $N$ for the LP relaxation given by $c_{\eps} \cdot \frac{\log N}{\log \log N}$
levels of the Sherali-Adams hierarchy.
\end{restatable}
%
Combining the above with the connection between Sherali-Adams gaps and extended formulations by Chan
\etal \cite{ChanLRS13} yields the following corollary. 
The improved tradeoff by Kothari \etal \cite{KothariMR16} gives a better exponent for $\log N$ than $3/2$.
\begin{restatable}{repcorollary}{sizecorollary}
\label{cor:size-lower-bound}
Let $f: \qary \to \B$ be any predicate. Let $\Phi_0$ be a $(c,s)$ integrality gap instance for basic
LP relaxation of \maxkcsp($f$). Then for every $\eps > 0$, there exists $c_{\eps} > 0$ 
such that for infinitely $N \in \N$,  there exist $(c-\eps, s+\eps)$ integrality gap instances of
size $N$,  for every linear extended formulation of size at most $\exp\inparen{c_{\eps} \cdot
  \frac{(\log N)^{3/.2}}{(\log \log N)^{1/2}}}$.
\end{restatable}
As an application of our methods, we also simplify and strengthen the approximation resistance
results for LPs proved by Khot \etal \cite{KhotTW14}. They studied predicates $f: \B^k \to \B$ and
provided a necessary and sufficient condition for the predicate to be {\deffont strongly
  approximation resistant} for the Sherali-Adams LP hierarchy. We say a predicate is strongly
approximation resistant if for all $\eps > 0$, it is hard to distinguish instances $\Phi$ for which
$\abs{\opt(\Phi) - \Ex{x \in \B^k}{f(x)}} \leq \eps$ from instances with $\opt(\Phi) \geq
1-\eps$. In the context of the Sherali-Adams hierarchy, they showed that when this condition is
satisfied, there exist instances $\Phi$ satisfying $\abs{\opt(\Phi) - \Ex{x \in
    \B^k}{f(x)}} \leq \eps$ and $\lpopt(\Phi) \geq 1-\eps$, where $\lpopt(\Phi)$ is the value of the
relaxation given by $O_{\eps}(\log \log n)$ levels of the Sherali-Adams hierarchy. We strengthen
their result (and provide a simpler proof) to prove the following.
\begin{theorem}\label{thm:ktw}
Let $f: \B^k \to \B$ be any predicate satisfying the condition for strong approximation resistance
for LPs, given by \cite{KhotTW14}. Then for every $\eps > 0$, there exists $c_{\eps} > 0$ 
such that infinitely many $N \in \N$, there exists an instance $\Phi$ of $\maxkcsp(f)$ of size $N$, satisfying
\[
\abs{\opt(\Phi) - \Ex{x \in \B^k}{f(x)}} ~\leq~ \eps
\qquad
\lpopt(\Phi) ~\geq~ 1-\eps \mcom
\]
where $\lpopt(\Phi)$ is the value of the relaxation given by $c_{\eps} \cdot \frac{\log N}{\log \log
  N}$ levels of the Sherali-Adams hierarchy.
\end{theorem} 
As before, the above theorem also yields a corollary for extended formulations.

\subsubsection*{Proof overview and techniques}
\paragraph{The gap instance.} 
The construction of our gap instances is inspired by the construction by Khot \etal
\cite{KhotTW14}. They gave a generic construction to prove integrality gaps for any approximation
resistant predicate (starting from certificates of hardness in form of certain ``vanishing
measures''), and we use similar ideas to give a construction which can start from a basic LP
integrality gap instance as a certificate, to produce a gap instance for a large number of
levels. This construction is discussed in \cref{sec:sagaps}.

Given an integrality gap instance $\Phi_0$ on $n_0$ variables, we treat this as a
``template'' (as in Raghavendra \cite{Raghavendra08}) and generate a random instance using
this. Concretely, we generate a new instance $\Phi$ on $n_0$ sets of $n$ variables each. To
generate a contraint, we sample a random constraint  $C_0 \in \Phi_0$, and pick a variable randomly
from each of the sets corresponding to variables in $C_0$. 
Thus, the instances generated are
$n_0$-partite random hypergraphs, with each edge being generated according to a specified ``type''
(indices of sets to chose vertices from). 
Previous instances of gap constructions for LP and SDP
hierarchies were (hyper)graphs generated according to the model ${\cal G}_{n,p}$. However,
properties of random ${\cal G}_{n,p}$ hypergraphs easily carry over to our instances, and we collect
these properties in \cref{sec:hypergraphs}.

The above construction ensures that if the instance $\Phi_0$ does not have an assignment satisfying
more than an $s$ fraction of the constraints, then $\opt(\Phi) \leq s+\eps$ with high probability.
Also, it is well-known that providing a good LP solution to the relaxation given by $t$ levels of the
Sherali-Adams hierarchy is equivalent to providing distributions $\dist_S$ on $\qary^S$ for all sets
of variables  $S$ with $\abs{S} \leq t$, such that the distributions are consistent restricted to
subsets \ie for all $S$ with $\abs{S} \leq t$ and all $T \subseteq S$, we have $\dist_{S|T} =
\dist_{T}$. 
Thus, in our case, we need to produce such consistent local distributions such that the expected
probability that a random constraint $C \in \Phi$ is satisfied by the local distribution on the set of
variables involved in $C$ (which we denote as $S_C$) is at least $c - \eps$.

\paragraph{Local distributions from local structure.}
Most works on integrality gaps for CSPs utilize the local structure of random
hypergraphs to produce such distributions. Since the girth of a sparse random hypergraph is
$\Omega(\log n)$, any induced subgraph on $o(\log n)$ vertices is simply a forest. In case the
induced (hyper)graph $G_S$ on a set $S$ is a \emph{tree}, there is an easy distribution to consider: simply
choose an arbitrary root and propagate down the tree by sampling each child conditioned on its
parent. It is also easy to see that for $T \subseteq S$, if the induced (hyper)graph  $G_T$ is a
\emph{subtree} of $G_S$, then the distributions $\dist_S$ and $\dist_T$ produced as above
are consistent.

The extension of this idea to forests requires some care. One can consider extending the
distribution to forests by propagating independently on each tree in the forest.
However, if for $T \subseteq S$ $G_T$ is a forest while $G_S$ is a tree, then a pair of vertices
disconnected in $G_T$ will have no correlation in $\dist_T$ but may be correlated in $\dist_S$.
This was handled, for example, in \cite{KhotTW14} by adding noise to the propagation and using a
large ball $B(S)$ around $S$ to define $\dist_S$. Then, if two vertices of $T$ are disconnected in
$B(T)$ but connected in $B(S)$, then they must be at a large distance from each other. Thus, because
of the noise, the correlation between them (which is zero in $\dist_T$) will be very small in
$\dist_S$. However, correcting approximate consistency to exact consistency incurs a cost which is
exponential in the number of levels (\ie the sizes of the sets), which is what limits the results in
\cite{KhotTW14, delaVegaK07} to $O(\log \log n)$ levels. This also makes the proof a bit more
involved since it requires a careful control of the errors in consistency.

\paragraph{Consistent partitioning schemes.}
We resolve the above consistency issue by first partitioning the given set $S$ into a set of
clusters, each of which have diameter $\Delta_H = o(\log n)$ in the underlying hypergraph $H$.
Since each cluster has bounded diameter, it becomes a tree when we add all the missing paths
between any two vertices in the cluster. We then propagate independently on each
cluster (augmented with the missing paths). This preserves the correlation between any two vertices
in the same cluster, even if the path between them was not originally present in $G_S$.

Of course, the above plan requires that the partition obtained for $T \subseteq S$, is consistent
with the restriction to $T$ of partition obtained for the set $S$. In fact, we construct
distributions over partitions $\inbraces{\calP_S}_{\abs{S} \leq t}$, which satisfy the consistency
property $\calP_{S | T} = \calP_T$. These distributions over partitions, which we call {\deffont
  consistent partitioning schemes}, are constructed in \cref{sec:decompositions}.

In addition to being consistent, we require that the partitioning scheme cuts only a small number of
edges in expectation, since these contribute to a loss in the LP objective. We remark that such
low-diameter decompositions (known as \emph{separating} and \emph{padded} decompositions) have been
used extensively in the theory metric embeddings (see \eg \cite{KrauthgamerLMN05} and the references
therein). The only additional requirement in our application is consistency.

We obtain the decompositions by proving the (easy) hypergraph extensions the results of Charikar,
Makarychev and Makarychev \cite{CharikarMM07}, who exhibit a metric which is similar to the shortest
path metric on graphs at small distances, and has the property that its restriction to any subset of
size at most $n^{\eps'}$ (for an appropriate $\eps' < 1$) 
is $\ell_2$ embeddable. This is proved in \cref{sec:hypergraphs}.
We then use these in \cref{sec:decompositions} to construct the consistent partitioning schemes as
described above, by applying  a result of Charikar \etal \cite{CharikarCGGP98} giving separating
decompositions for finite subsets of $\ell_2$.

We remark that it is the consistency requirement of the partitioning procedure that limits our results
to $O\inparen{\frac{\log n}{\log \log n}}$ levels. The separation probability in the decomposition
procedure grows with the dimension of the $\ell_2$ embedding, while (to the best of our knowledge)
dimension reduction procedures seem to break consistency.


%% file: prelims.tex
We use $[n]$ to denote the set $\inbraces{1, \ldots, n}$. The only exception is $\qary$, where we
overload this notation to denote the set $\inbraces{0, \ldots, q-1}$, which corresponds to the the
alphabet for the Constraint Satisfaction Problem under consideration. We use $\dist_S$ and $\calP_S$
to denote probability distributions over (assignments to or partitions of) a set $S$. For $T
\subseteq S$, the notation $\dist_{S|T}$ is used to denote the restriction (marginal) of the
distribution $\dist_{S}$ to the set $T$ (and similarly for $\calP_{S|T}$).

\subsection{Constraint Satisfaction Problems}\label[section]{sec:csps}
\begin{definition}
Let $\qary$ denote the set $\inbraces{0, \ldots, q-1}$. 
For a predicate $f : \qary^k \rightarrow \{0,1\}$, an instance
$\Phi$ of \maxkcspq$(f)$ consists of a set of variables $\{x_1,\ldots,x_n\}$ and a set of
constraints $C_1, \ldots, C_m$ where each constraint $C_i$ is over a $k$-tuple of variables
$\{x_{i_1}, \ldots, x_{i_k}\}$ and is of the form
\[
C_i ~\equiv~ f(x_{i_1} + b_{i_1}, \ldots, x_{i_k} + b_{i_k})
\]
for some $b_{i_1}, \ldots, b_{i_k} \in \qary$, where the addition is modulo $q$.
For an assignment $\sigma: \{x_1,\ldots,x_n\} \mapsto \qary$, let $\sat(\sigma)$ denote the fraction of
constraints satisfied by $\sigma$. 
The maximum fraction of constraints that can be simultaneously satisfied is denoted by $\opt(\Phi)$,
i.e.
\[ 
\opt(\Phi) = \max_{\sigma: \{x_1,\ldots,x_n\} \mapsto \qary}  \sat(\sigma). 
\]
\end{definition}

For a constraint $C$ of the above form, we use $x_C$ to denote the tuple of variables $(x_{i_1},
\ldots, x_{i_k})$ and $b_C$ to denote the tuple $(b_{i_1}, \ldots, b_{i_k})$. We then write
the constraint as $f(x_C + b_C)$. We also denote by $S_C$ the set of indices
$\{i_1,\ldots, i_k\}$ of the variables participating in the constraint $C$.

\subsection{The LP Relaxations for Constraint Satisfaction Problems}\label[section]{sec:relaxations}
Below we present various LP relaxations for the \maxkcspq$(f)$ problem that are relevant in this paper.

We start with the level-$t$ Sherali-Adams relaxation.
The intuition behind it is the following. Note that an integer
solution to the problem can be given by an assignment $\sigma: [n] \to \qary$.  Using this, we can
define $\{0,1\}$-valued variables $\vartwo{S}{\alpha}$ for each $S \subseteq [n], 1 \leq |S| \leq t$
and $\alpha \in \qary^S$,
with the intended solution $\vartwo{S}{\alpha} = 1$ if $\sigma(S) = \alpha$ and 0 otherwise.
We also introduce a variable $\vartwoempty$, which equals 1.
We relax the integer program and allow variables to take real values in $[0,1]$.
Now the variables $\{ \vartwo{S}{\alpha} \}_{\alpha \in \qary^{S}}$ give a probability distribution $\dist_S$
over assignments to $S$. We can enforce consistency between these {\emph local} distributions
by requiring that for $T \subseteq S$, the distribution over assignments to $S$, when marginalized
to $T$, is precisely the distribution over assignments to $T$ \ie $\dist_{S | T} = \dist_T$.
 The relaxation is shown in Figure \ref{fig:SA-lp}.

\begin{figure}[ht]
\hrule
\vline
\begin{minipage}[t]{0.99\linewidth}
\vspace{-5 pt}
{\small
\begin{align*}
\mbox{maximize} &~~
\Ex{C \in \Phi}{\sum_{\alpha \in \qary^k} f(\alpha \cdot b_C) \cdot \vartwo{S_C}{\alpha} }&  \\
\mbox{subject to} \\ 
%
%
%
%
\sum_{\alpha \in \qary^{S} \atop \alpha|_T = \beta}\vartwo{S}{\alpha} &~=~
\vartwo{T}{\beta}  &\forall T \subseteq S \subseteq [n], |S|\leq t, ~\forall \beta \in \qary^T \\
 \vartwo{S}{\alpha} &~\geq~ 0
  & \forall S \subseteq [n], |S| \leq t,  ~\forall \alpha \in \qary^{S} \\
   \vartwoempty & ~=~  1    \\
\end{align*}}
\vspace{-14 pt}
\end{minipage}
\hfill\vline
\hrule
\caption{Level-$t$ Sherali-Adams LP for \maxkcspq($f$)}
\label{fig:SA-lp}
\end{figure}

The basic LP relaxation is a reduced form of the above relaxation where only
those variables $\vartwo{S}{\alpha}$ are included for which $S  = S_C$ is the set of CSP variables
for some constraint $C$. The consistency constraints are included
only for singleton subsets of the sets $S_C$.
Note that the all the constraints for the
basic LP are implied by the relaxation obtained by level $k$ of the Sherali-Adams hierarchy.
\begin{figure}[ht]
\hrule
\vline
\begin{minipage}[t]{0.99\linewidth}
\vspace{-5 pt}
{\small
\begin{align*}
\mbox{maximize} &~~
\Ex{C \in \Phi}{\sum_{\alpha \in \qary^k} f(\alpha + b_C) \cdot \vartwo{S_C}{\alpha} }&  \\
\mbox{subject to}\\ 
\sum_{j \in \qary}\vartwo{i}{b} &~=~ 1 &\forall i \in [n] \\
\sum_{\alpha \in \qary^{S_C} \atop \alpha(i) = b}\vartwo{S_C}{\alpha} &~=~
\vartwo{i}{b} &\forall C \in \Phi, i \in S_C, b \in \qary\\
 \vartwo{S_C}{\alpha} &~\geq~ 0
  & \forall C \in \Phi, ~\forall \alpha \in \qary^{S_C}
\end{align*}}
\vspace{-14 pt}
\end{minipage}
\hfill\vline
\hrule
\caption{Basic LP relaxation for \maxkcspq($f$)}
\label[figure]{fig:basic-lp}
\end{figure}

For an LP/SDP relaxation of \maxkcspq, and for a given instance $\Phi$ of the problem, we denote by
$\sdpopt(\Phi)$ the LP/SDP (fractional) optimum.
%
%
A relaxation is said to have a $(c, s)$-integrality gap if there exists a
CSP instance $\Phi$ such that $\sdpopt(\Phi) \geq c$ and $\opt(\Phi) < s$.

\subsection{Hypergraphs}
An instance $\Phi$ of \maxkcsp defines a natural associated hypergraph $H = (V,E)$ with $V$ being
the set of variables in $\Phi$ and $E$ containing one $k$-hyperedge for every constraint $C \in \Phi$.
We remind the reader of the familiar notions of degree, paths, and cycles for the case of 
($k$-uniform) hypergraphs:
\begin{definition}
  Let $H=(V,E)$ be a hypergraph. 
\begin{itemize}
\item For a vertex $v\in V$, the {\deffont degree} of the vertex $v$ is defined to be the number of
  distinct hyperedges containing it.
\item   A {\deffont simple path} $P$ is a finite alternate sequence of distinct vertices and
  distinct edges starting and ending at vertices, \ie, $P=v_1,e_1,v_2,\ldots ,v_{l},e_{l},v_{l+1}$,
  where $v_i \in V ~\forall i \in [l+1]$ and $e_i\in E~ \forall i\in [l]$.  Furthermore, $e_i$
  contains $v_i,v_{i+1}$ for each $i$. Here $l$ is called the {\deffont length} of the path $P$. 
  All paths discussed in this paper will be simple paths.
\item A sequence $\calC = (v_1,e_1,v_2,\ldots, v_l,e_l,v_1)$ is called a cycle of length
  $l$ if the initial segment $v_1,e_1,\ldots, v_l$ is a (simple) path, $e_{l+1}\ne e_i$ for all $i\in [l]$,
  and $v_1\in e_l$. For a path $P$ (or cycle $\calC$), we use $\Vtx(P)$ (or $\Vtx(\calC)$) to 
  denote the set of
  vertices all the vertices that occurs in the edges, \ie the set $\{v\suchthat (\exists i\in[h])
  (v\in e_i)\}$, where $e_1,\ldots,e_h$ are the edges included in $P$ (or $\calC$).
\item For a given hypergraph $H$, the length of the smallest cycle in $H$ is called the
  {\deffont girth} of $H$.
\end{itemize}
\end{definition}
To observe the difference the notions of cycle in graphs and hypergraphs, it is instructive to
consider the following example: let $u,v$ be two distinct vertices in a $k$-uniform hypergraph for
$k \geq 3$, and let $e_1,e_2$ be two distinct hyperedges both containing $u$ and $v$. Then
$u,e_1,v,e_2,u$ is a cycle of length $2$, which cannot occur in a graph.
 
We shall also need the following notion of the \emph{closure} of a set $S \subseteq V$ in a given
hypergraph $H$, defined by \cite{CharikarMM09} for the case of graphs. A stronger notion of closure
was also considered by \cite{BarakCK15}.
\begin{definition}\label[definition]{def:closure}
  For a given hypergraph $H$ and $R \in \N$, and a set $S\subseteq \Vtx(H)$, we denote by $\clR(S)$ 
  the {\deffont $R$-closure} of $S$ obtained by adding all the vertices in all the paths of length at most $R$
  connecting two vertices of $S$, \ie
  \[\clR(S)=S\cup \bigcup_{\substack{P: P\text{ is a path in H} \\ P\text{ connects }u,v\in S \\
  \abs{P}\le R}}\Vtx(P)\mper\]
  For ease of notation, we use $\cl(S)$ to denote $\cl_1(S)$.
\end{definition}


%% file: hypergraphs.tex
\newcommand{\1}{\ensuremath{\mathbb{1}}}
We collect here various properties of the hypergraphs corresponding to our integrality gap
instances. 
The gap instances we generate contain several disjoint collections of variables. Each
constraint in the instance has a specified ``type'', which specifies which of the collections each
of the participating $k$ variables much be sampled from. The constraint is generated by randomly 
sampling each of the $k$ variables, from the collections specified by its type. This is captured by
the generative model described below. 

In the model below and in the construction of the gap instance, the parameter $n_0$ should be
thought of as constant, while the parameters $n$ and $m$ should be though of a growing to infinity. 
We will choose $m = \gamma \cdot n$ for $\gamma = O_{k,q}(1)$.
\begin{definition}
Let $n_0, k \in \N$ with $k \geq 2$. Let $m, n > 0$ and let $\Gamma$ be a distribution on
$\insquare{n_0}^k$. We define a distribution $\Hyp$ on $k$-uniform $n_0$-partite hypergraphs with $m$
edges and $N=n_0 \cdot n$ vertices, divided in $n_0$ sets $X_1, \ldots, X_{n_0}$ of size $n$ each. A
random hypergraph $H \sim \Hyp$ is generated by sampling $m$ random hyperedges independently as
follows:
\begin{itemize}
\item Sample a random type $(i_1, \ldots, i_k) \in [n_0]^k$ from the distribution $\Gamma$.
\item For all $j \in [k]$, sample $v_{i_j}$ independently and uniformly in $X_{i_j}$.
\item Add the edge $e_i = \inbraces{v_{i_1}, \ldots, v_{i_k}}$ to $H$.
\end{itemize}
\end{definition}
Note that as specified above, the model may generate a multi-hypergraph. However, the number of such
repeated edges is likely to be small, and we will bound these, and in fact the number of cycles of
size $o(\log n)$ in \cref{lem:cycle-count}.

We will study the following metrics (similar to the ones defined in \cite{CharikarMM07}) in this section:
\begin{definition}
  \label{def:metric-dmu}
  Given a hypergraph $H$ with vertex set $V$, we define two metrics $d^H_\mu(\cdot\xspace,\cdot), \rho^H_\mu
  (\cdot\xspace,\cdot)$ on $V$ as
\[
    d^H_\mu(u,v) ~\coloneqq~ 1-(1-\mu)^{2\cdot d_H(u,v)}
     \qquad \text{ and } \qquad \rho^H_\mu(u,v) ~\coloneqq~
     \sqrt{\frac{2\cdot d^H_\mu(u,v)+ \mu}{1+\mu}}\mcom
\]
  for $u\ne v$, where $d_H(\cdot\xspace,\cdot)$ denotes the shortest path distance in $H$.
\end{definition}

The goal of this section is to prove the following result about the local $\ell_2$-embeddability of
the metric $\rho_{\mu}$. The proof of the theorem heavily uses results proved in \cite{AroraBLT06}
and \cite{CharikarMM09}.
\begin{theorem}
  \label{thm:locally-l2}
  Let $H'\sim \Hyp$ with $m=\gamma\cdot n$ edges and let $\eps>0$. Then for large
  enough $n$, with high probability (at least $1-\eps$, over the choice of $H'$),
  there exists $\delta>0$, constant $c = c(k,\gamma,n_0,\eps)$, $\theta = \theta(k,\gamma,n_0,\eps)$ 
  and a subhypergraph $H\subset H'$ with $V(H) = V(H')$ satisfying the following:
\begin{itemize}
\item H has girth $\ttg\ge \delta\cdot \log n$.
\item $\abs{\Edg(H') \setminus\Edg(H)}\le \eps\cdot m$.
\item For all $t \leq n^{\theta}$, for $\mu \geq c \cdot \frac{\log t+\log\log n}{\log n}$, 
  for all $S\subseteq \Vtx(H')$ with
  $\abs{S} \leq t$, the metric $\rho^H_{\mu}$ restricted to $S$ is isometrically embeddable
  into the unit sphere in $\ell_2$,
\end{itemize}
%
%
\end{theorem}
%
\mtnote{Changed the metric to $\rho^H_{\mu}$ to clarify that we use the distances in $H$ and not the
original $H'$.}
%


To prove the above theorem, we will use the local structure of random (hyper)graphs.
We first prove that with high probability random hypergraphs (sampled from $\Hyp$) can be
modified by removing a few edges to a hypergraph whose girth is $\Omega(\log n)$ and the
degree of the resulting hypergraph is bounded.
The following lemma shows a possible
trade-off between the degree of the hypergraph vs the number of edges required to be
removed.
\begin{lemma}\label[lemma]{lem:low-degree}
  Let $H'\sim \Hyp$ be a random hypergraph with $m=\gamma\cdot n$ edges. Then for any
  $\eps>0$, with probability $1 - \eps$ the exists a sub-hypergraph
  $H$ with $V(H) = V(H')$ such that $\forall u \in V(H)$, $\deg_H(u) \leq
  100\cdot\log\inparen{\frac{n_0}{\eps}}\cdot k \cdot \gamma$ and $\abs{\Edg(H')\setminus
  \Edg(H)}\le \eps\cdot m$.
\end{lemma}
\begin{proof}
  By linearity of expectation, the expected degree of any vertex $v$ in $H'$ is at most
  $k\cdot\gamma$,. Let $D = 100\cdot\log\inparen{\frac{n_0}{\eps}}\cdot k\cdot \gamma$. Let $S$ be the
  set of all vertices $u$ such that $\deg_{H'}(u) > D$. Let $E_S$ be the set of all hyperedges with
  one vertex in $S$. We shall take $E(H) = E(H') \setminus E_S$. 
Note that for any $u \in V(H')$, $\prob{u \in S} = \prob{\deg_{H'}(u) \geq D} \leq \exp(-D/4)$ by a
Chernoff-Hoeffding bound. We use this to bound the expected number of edges deleted.
\begin{align*}
\ex{E_S} 
~\leq~ \sum_{u \in V(H')} \ex{\deg(u) \cdot \indicator{u \in S}} 
&~=~ \sum_{u \in V(H')} \ex{\deg(u) ~\mid~ u \in S} \cdot \prob{u \in S} \\
&~\leq~ \sum_{u \in V(H')} \ex{\deg(u) ~\mid~ u \in S} \cdot \exp\inparen{-D/4} \\
&~\leq~ \sum_{u \in V(H')} \inparen{D+k\gamma} \cdot \exp\inparen{-D/4} \\
&~\leq~ (n \cdot n_0) \cdot 2D \cdot \exp\inparen{-D/4} \mper
\end{align*}
The penultimate inequality uses the independence of the hyper-edges in the generation process, which gives
$\ex{\deg_{H'}(u) \mid \deg_{H'}(u) \geq D} ~\leq~ D + \ex{\deg_{H'}(u)}$. From our choice of the
parameter $D$, we get that $\ex{E_S} \leq \eps^2 \cdot \gamma \cdot n = \eps^2 \cdot m$. Thus, the
number of edges deleted is at most $\eps \cdot m$ with probability at least $1-\eps$.
\end{proof}
The following lemma shows that the expected number of small cycles in random hypergraph is small.
\begin{lemma} \label[lemma]{lem:cycle-count}
Let $H \sim \Hyp$ be a random hypergraph and for $l \geq 2$, let $Z_{l}(H)$ denote the number of
cycles of length at most $l$ in $H$. For $m, n$ and $k$ such that $k^2 \cdot (m/n) > 1$, we have
  \[
    \Ex{H \sim \Hyp}{Z_{l}(H)} \le \inparen{k^2 \cdot \frac{m}{n}}^{2l} \mper
  \]
\end{lemma}
\begin{proof}
  Let the vertices of $H$ correspond to the set $[n_0]\times [n]$.
  Suppose we contract the set of $[n_0]\times\{j\}$ vertices into a single
  vertex $j \in [n]$ to get a random multi-hypergraph $H'$ on vertex set $[n]$.
  An equivalent way to view the sampling to $H'$ is: for each $i\in [m]$, the
  $i$-th hyperedge $e_i$ of $H'$ is sampled by independently sampling $k$ vertices (with
  replacement) uniformly at random from $[n]$. Note
  that the sampling of $H'$ is independent of $\Gamma$ in the definition of
  $\Hyp$.
  Clearly, a cycle of length at most $l$ in $H$ produces a cycle of length at most $l$ in
  $H'$. Hence, suffices to bound the expected number of cycles in $H'$

  Given any pair $(u',v')$ of vertices of $H'$, for $u' \ne v'$, the probability of the
  pair $(u',v')$ belonging together in some edge of $H'$ is at most $\frac{mk^2}{n^2}$.
  Consider a given $h$-tuple  of vertices $\bfu=(u_{i_1},\cdots u_{i_h})$. Note that we
  require that edges participating in a cycle be distinct. So, the probability that $\bfu$
  is part of a cycle in $H'$, \ie there exists distinct edges $e_j\in H'$ for $j\in [h]$
  such that $u_{i_j},u_{i_{j+1}}\in e_j$ for $j\in [h-1]$, and $u_{i_1},u_{i_h}\in e_{h}$
  is at most $\inparen{\frac{mk^2}{n^2}}^{h}$. As a result, expected number of
  cycles of length $h$ in $H'$ is bounded above by:
  \[
    {n \choose h} \inparen{\frac{mk^2}{n^2}}^{h} \le n^h \inparen{\frac{mk^2}{n^2}}^{h} =
    \inparen{k^2\cdot\frac{m}{n}}^{h}
  \]
  From the geometric form of the bound, it follows that expected number of cycles of length
  at most $l$ in $H'$ is at most $\frac{\inparen{k^2\cdot\frac{m}{n}}^{l+1}}{
    \inparen{k^2\cdot\frac{m}{n}}-1}<\inparen{k^2\cdot\frac{m}{n}}^{2l}\mper$
\end{proof}
Using the above lemma, it is easy to show that one can remove all small cycles in a random
hypergraph by deleting only a small number of edges.
\begin{corollary}\label[corollary]{cor:large-girth}
  Let $H \sim \Hyp$ be a random hypergraph with $m = \gamma \cdot n$ for $\gamma>1$ and $k\ge 2$.
  Then, there exists $\delta=\delta(\gamma) > 0$ such that with probability $1-n^{-1/6}$, all cycles
  of length at most $\delta \cdot \log n$ in $H$ can be removed by deleting at most $n^{2/3}$ edges.
\end{corollary}
\begin{proof}
  As above, let $Z_l$ denote the number of cycles of length at most $l$. With the choice of $m,n,$ and $k$,
  we have $k^2\cdot\frac{m}{n}\ge 2$. By \cref{lem:cycle-count}, 
  $\ex{Z_l} \leq \inparen{k^2 \cdot \frac{m}{n}}^{2l}$. Since $m = \gamma \cdot n$, there exists a
  $\ttg = \delta \cdot \log n$ such that  $\ex{Z_l} \leq \sqrt{n}$. By Markov's inequality,
  $\prob{Z_l \geq n^{2/3}} \leq n^{-1/6}$. Thus, with probability $1-n^{-1/6}$, one can remove
  all cycles of length at most $\delta \cdot \log n$ by deleting at most $n^{2/3}$ edges.
\end{proof}
One can also extend the analysis in \cite{AroraBLT06} to show that the hypergraphs are locally
sparse \ie the number of edges contained in a small set of vertices is small. For a hypergraph $H$
and a set $S \subseteq \Vtx(H)$, we use $E(S)$ to denote the edges contained in the set $S$. 
\begin{definition}
We say that $S \subseteq V(H)$ is $\eta$-sparse if $\abs{E(S)} \leq \frac{\abs{S}}{k-1-\eta}$. We
call an $k$-uniform hypergraph $H$ on $N$ vertices to be {\deffont $(\tau,\eta)$-sparse} if all
subsets $S\subset \Vtx(H), \abs{S}\le \tau \cdot \abs{\Vtx(H)}$, $S$ is $\eta$-sparse. 
We call $H$ to be $\eta$-sparse if it is $(1,\eta)$-sparse, \ie all subsets of vertices of $H$ are
sparse.
\end{definition}
We note here that while this notion of sparsity is a generalization of that considered in
\cite{AroraBLT06}, it is also identical to the notions of expansion considered in works in proof
complexity (see \eg \cite{BS01}) and later in works on integrality gaps
\cite{AroraAT05, BenabbasGMT12,BarakCK15}.
We prove that random hypergraphs generated with our model are locally sparse:
%
%
%
\begin{restatable}{replemma}{locallysparse}
\label{lem:locally-sparse}
  Let $\eta < 1/4$ and $m=\gamma\cdot n$ for $\gamma > 1$. Then
  for $\tau \le \frac{1}{n_0} \cdot \inparen{\frac{1}{e \cdot k^{3k} \cdot
        \gamma}}^{1/\eta}$ the following holds:
  \[
    \Prob{H \sim \Hyp}{H~\text{is not}~(\tau,\eta)\text{-sparse}}
    ~\leq~ 3 \cdot \inparen{\frac{k^{3k} \cdot \gamma}{n^{\eta/4}}}^{1/k} \mper
  \]
\end{restatable}
 We note that we will require the sparsity $\eta$ to be $O_{k,\gamma}(1/(\log n))$. This gives
 sparsity only for sublinear size sets, as compared to sets of size $\Omega(n)$ in previous works
 where $\eta$ is a constant. 
For the proof of the lemma, we follow an approach similar to that of \cref{lem:cycle-count}: we collapse
the vertices of $H$ of the form $[n_0]\times\{j\}$ to vertex $j\in [n]$ to construct $H'$, and thus reducing
the problem to random multi-hypergraph form a random multipartite hypergraph.
The rest proof of the lemma is along the lines of several known proofs \cite{AroraAT05,
  BenabbasGMT12} and we defer the details to \cref{sec:appendix}.

%
%
%
%
Charikar \etal \cite{CharikarMM07} prove an analogue of \cref{thm:locally-l2} for metrics defined on
locally-sparse graphs. In fact, they  use a consequence of sparsity, which they call $\ell$-path
decomposability. To this end, we define the \emph{incidence graph}\footnote{This is the same notion
  as the constraint-variable graph considered in various works on lower bounds for CSPs.} 
associated with a hypergraph, on which we will apply their result.
\begin{definition}
Let $H = (V(H),E(H))$ be a $k$-uniform hypergraph. We define its {\deffont incidence graph} as the
bipartite graph $G_H$ defined on vertex sets $V(H)$ and $E(H)$, and edge set $\calE$ defined as
\[
\calE ~\defeq~ \inbraces{(v,e) ~\mid~ v \in V(H), ~e \in E(H), ~v \in e} \mper
\]
\end{definition}
Note that for any $u,v \in V(H)$, we have $d_{G_H}(u,v) = 2 \cdot d_H(u,v)$.
We prove that for a locally sparse hypergraph $H$, its incidence graph $G_H$ is also locally sparse.
\begin{lemma}
  \label[lemma]{lem:sparsity-variable-constraint}
  Let $H$ be a $k$-uniform $(\tau,\eta)$-sparse hypergraph on $N$ vertices with $m=\gamma\cdot n$
  hyperedges. Then the incidence graph $G_H$ is $(\tau',\eta')$ sparse for $\tau' =
  \nfrac{\tau}{k\cdot(1+\gamma)}$ and $\eta' = \nfrac{\eta}{(1+\eta)}$.
\end{lemma}
\begin{proof}
  Let $\tau'=\nfrac{\tau}{k\cdot(1+\gamma)}$ and let $G_H$ be the incidence graph with
  $N+m = (1+\gamma) \cdot N$ vertices. Let $G'$ be is the densest subgraph of $G_H$, among all subgraphs
  of size at most $\tau' \cdot (N+m)$. 
Let the vertex set of $G'$ be $V' \cup E'$ where $V' \subseteq V(H)$ and $E' \subseteq E(H)$, and
let the edge-set be $\calE'$. There cannot be any isolated vertices in $G'$ since removing those will only
increase the density.

Let $S \subseteq V(H)$ be the set of all vertices contained in all edges in $E'$ \ie 
$S ~\defeq~ \inbraces{v \in V(H) ~\mid~ \exists e \in E'~\text{s.t.}~ v \in e}$. Note that $V'
\subseteq S$, since there are no isolated vertices,  
and $E' \subseteq E(S)$, where $E(S)$ denotes the set of hyperedges contained in $S$.

By our choice of parameters, $\abs{S} \leq k \cdot \abs{E'} \leq k \cdot \tau' \cdot (N+m) \leq \tau
\cdot N$.
Thus, using the sparsity of $H$, we have
\[
\abs{E'} ~\leq~ \abs{E(S)} ~\leq~ \frac{\abs{S}}{k-1-\eta} \mper
\]
Also, since each hyperedge of $E'$ can include at most $k$ vertices in $S$, and since each edge in
$\calE'$ is incident on a vertex in $V'$, we have
\[
\abs{S} - \abs{V'} ~\leq~ k \cdot \abs{E'} - \abs{\calE'} \mper
\]
Combining the two inequalities gives
\[
(k-1-\eta) \cdot \abs{E'} ~\leq~ \abs{V'} + k \cdot \abs{E'} - \abs{\calE'}
\quad \implies \quad
\abs{\calE'} ~\leq~ (1+\eta) \cdot \abs{E'} + \abs{V'}  \mper
\]
Hence, we get that $\abs{\calE'} \leq \frac{\abs{V'} + \abs{E'}}{(1-\eta')}$ for $\eta' = \frac{\eta}{(1+\eta)}$.
\end{proof}
Charikar \etal \cite{CharikarMM07} defined the following structural property of a graph.
\begin{definition}[\cite{CharikarMM07}]\label[definition]{def:l-path-decmpsble} 
  A graph $G$ is {\sf $\ell$-path decomposable} if every $2$-connected subgraph $G'$ of $G$,
  such that $G'$ is not an edge, contains a path of length $\ell$ such that every vertex of the
  path has degree at most $2$ in $G'$.
\end{definition}
%
%
%
The above property was also implicitly used by Arora \etal (\cite{AroraBLT06}), who proved the
following (see Lemma 2.12 in \cite{AroraBLT06}):
\begin{lemma}\label[lemma]{lem:path-decomposable}
  Let $\ell>0$ be an integer and $0< \eta < \frac{1}{3\ell-1} < 1$.
  Let $G$ be a $\eta$-sparse graph with girth $\girth>\ell$.
  Then $G$ is $\ell$-path decomposable.
\end{lemma}
\mtnote{Check if $\girth > \ell$ is enough or one needs $\girth > 2\ell$.}
Recall that we defined the metrics $\dmu$ and $\rho_\mu$ on $H$ as (for $u\ne v$) :
\[
    d^H_\mu(u,v) ~\coloneqq~ 1-(1-\mu)^{2\cdot d_H(u,v)}
     \qquad \text{ and } \qquad \rho^H_\mu(u,v) ~\coloneqq~
     \sqrt{\frac{2\cdot d^H_\mu(u,v)+ \mu}{1+\mu}}\mcom
\]
For a graph $G$, we define the following two metrics, for $u\ne v$:
\[
   d_\mu^G(u,v) ~\coloneqq~ 1-(-1)^{d_G(u,v)}(1-\mu)^{d_G(u,v)}
  \quad \text{ and }  \quad
\rho_\mu^G(u,v) ~\coloneqq~
  \sqrt{\frac{2\cdot d_\mu^G(u,v)+ \mu}{1+\mu}}\mper
\]
We note that if $H$ is a hypergraph and $G_H$ is its incidence graph, then the metrics 
$d_\mu^{G_H}$ and $\rho_\mu^{G_H}$ restricted to $V(H)$, coincide with the metrics $d_\mu$ and
$\rho_\mu$ defined on $H$.
Charikar \etal proved the following theorem (see Theorem $5.2$) in \cite{CharikarMM09}.
\begin{theorem}[\cite{CharikarMM09}]\label[theorem]{thm:l2-graph-embedding}
  Let $G$ be a graph on $n'$ vertices with maximum degree $D$. Let
  $t< \sqrt{n'}$ and $\ell > 0$ be such that for $t'= D^{\ell+1}\cdot t$, every
  subgraph of $G$ on at most $t'$ vertices is $\ell$-path decomposable. Also, let $\mu$, $t$ and
  $\ell$ satisfy the relation $(1-\mu)^{\ell/9} \leq \frac{\mu}{2(t+1)}$.
Then for every subset $S$ of at most $t$ vertices there exists
  a mapping $\psi_S$ from $S$ to \emph{unit sphere} in $\ell_2$ such that all $u,v \in S$:
  \[
    \norm{\psi_S(u)-\psi_S(v)}_2 ~=~ \rho_\mu^G(u,v) \mper
  \]
\end{theorem}
We use this theorem to prove the main theorem of the section.
\begin{proofof}{of \cref{thm:locally-l2}}
  Let $H'\sim \Hyp$ with $m=\gamma\cdot n$ hyperedges and $N=n_0\cdot n$
  vertices. Given $\eps>0$, from \cref{lem:low-degree} we have that with
  high probability at least $1-\eps/2$, there exists $H_1$ such that 
  the maximum degree of $H_1$ is at most $D=100\cdot\log\inparen{\frac{2n_0}{\eps}}
  \cdot k \cdot \gamma$ with $\abs{\Edg(H')\setminus\Edg(H_1)}\le (\eps/2)\cdot m$.

Using \cref{cor:large-girth} we also have that there exists $\delta > 0$, such that with probability
at least $1-\eps/4$ (for large enough $n$) $H'$ has a sub-hypergraph $H_2$ with $\ttg \geq
\delta\cdot\log n$ and $\abs{\Edg(H') \setminus\Edg(H_2)}\le (\eps/4) \cdot m$.
By \cref{lem:locally-sparse}, there exists $\eta = \Omega_{n_0,k,\gamma,\eps}(1/(\log n))$ such
that $H'$ is $(\tau,\eta)$-sparse with probability at least $1-\eps/4$, for $\tau \geq n^{-1/4}$.
%
  
Hence with probability $1-\eps$, we have that $H = (V(H'), E(H_1) \cap E(H_2))$ satisfies:
  \begin{itemize}
    \item Degree of $H$ is bounded above by $D$.
    \item $H$ is $(\tau,\eta)$-sparse (for $\tau \geq n^{-1/4}$ and $\eta =
      \Omega_{n_0,k,\gamma,\eps}(1/(\log n)$).
    \item Girth of $H$ is at least $\ttg > \delta\cdot\log n$.
    \item $\abs{\Edg(H')\setminus\Edg(H)}\le \eps\cdot m$.
  \end{itemize}
We now show that the metric $\rho^H_{\mu}$ is locally $\ell_2$ embeddable.

Let $G=G_H$ be the incidence graph for the hypergraph $H$. Note that $N \leq \abs{\Vtx(G)} \leq 
  N\cdot (1+\gamma)$ and degree of $G$ is also bounded by $D$. 
  Since a cycle in $G$ is also a cycle in $H$, the girth of $G$ is at also least $\girth \geq \delta
  \cdot \log n$.

  By \cref{lem:sparsity-variable-constraint}, we have $G$ is $(\frac{\tau}{k (1+\gamma)},
  \frac{\eta}{1+\eta})$-sparse. 
By \cref{lem:path-decomposable}, any subgraph of $G$ on at most $\frac{\tau}{k(1+\gamma)}
\cdot (N+m)$ vertices is $\ell$-path decomposable for any $\ell \leq \min\{\girth,
  1/(4\eta)\}$. 
Since $D = 100 \cdot k\gamma \cdot \log(2n_0/\eps)$, there exists $\ell_0 =
\Omega_{k,\gamma,n_0,\eps}(\log n)$ such that $D^{\ell_0+1} \leq n^{1/6}$. We choose $\ell =
\min\inbraces{\girth, 1/(4\eta), \ell_0}$.

Let $\mu_0$ be the smallest $\mu$ such that $\exp\inparen{-\mu\ell/9} \leq \frac{\mu}{2(t+1)}$ (note
that $\frac{1}{\mu} \cdot \exp\inparen{-\mu\ell/9}$ is decreasing in $\mu$). 
Since we must have $\mu \geq 1/\ell$, there exists a $\mu_0$ satisfying
\[
\mu_0 ~\leq~ \frac{9}{\ell} \cdot \inparen{\ln(2(t+1)) + \ln \ell} \mper
\]
From our choice of $\ell$, there exist constants $c = c(k,\gamma,n_0,\eps)$ and $\theta =
\theta(k,\gamma,n_0,\eps) < 1/2$ such that $\mu_0 \leq c \cdot \frac{\log t + \log \log n}{\log n}
< 1$ when $t \leq n^{\theta}$. 
Then, for any $\mu \in [\mu_0,1)$, we have $(1-\mu)^{\ell/9} \leq \exp(-\mu\ell/9) \leq
\frac{\mu}{2(t+1)}$.

We can now apply \cref{thm:l2-graph-embedding} to construct the embedding. 
Given any subset $S$ of $\Vtx(H)$ of size at most $t\le n^{\theta}$, note that $S$ is also a subset
of $\Vtx(G)$. Moreover, we have $t \leq n^{\theta} \leq (N+m)^{1/2}$. Also, we have $t \cdot
D^{\ell+1} \leq n^{1/2} \cdot n^{1/6} = n^{2/3} \leq \frac{\tau}{k(\gamma+1)} \cdot (N+m)$. Thus,
any subgraph of $G$ on $t \cdot D^{\ell+1}$ vertices is $\ell$-path decomposable. 

Thus, when $\mu \geq \mu_0$, by \cref{thm:l2-graph-embedding} 
there exists a mapping $\psi_S$ from $S$ to the unit sphere, such that for all $u,v \in S$, we have
\[
\norm{\psi_S(u)-\psi_S(v)}_2 ~=~ \rho^G_{\mu}(u,v) ~=~ \rho^H_{\mu}(u,v) \mcom
\]
where the last equality uses the fact that for all $u,v \in \Vtx(H)$, $\rho^H_{\mu}(u,v) =
\rho^G_{\mu}(u,v)$ since $d_{G}(u,v) = 2 \cdot d_H(u,v)$.
\end{proofof}

%% file: decompositions.tex


We will construct the Sherali-Adams solution by partitioning the given subset of vertices in to trees,
and then creating a natural distribution over satisfying assignments on trees. We define below the
kind of partitions we need.
\begin{definition}\label[definition]{def:partioning-scheme}
Let $X$ be a finite set. For a set $S$, let $\calP_S$ denote a distribution over partitions of $S$. For $T \subseteq S$, let $\calP_{S|T}$ be the distribution over partitions of $T$ obtained by restricting the partitions in $\calP_S$ to the set $T$.
We say that a collection of distributions $\inbraces{\calP_S}_{\abs{S} \leq t}$ forms a {\deffont consistent partitioning scheme of order} $t$, if 
\[
\forall S \subseteq X, \abs{S} \leq t ~\text{and}~ \forall T \subseteq S 
\qquad \calP_T = \calP_{S|T} \mper
\]
\end{definition}
In addition to being consistent as described above, we also require the distributions to have small probability of cutting the edges for the hypergraphs corresponding to our CSP instances. We define this property below.
\begin{definition}\label[definition]{def:sparse-partitioning}
Let $H = (V,E)$ be a $k$-uniform hypergraph. Let $\inbraces{\calP_S}_{\abs{S} \leq t}$ be a consistent
partitioning scheme of order $t$ for the vertex set $V$, with $t \geq k$. We say the scheme
$\inbraces{\calP_S}_{\abs{S} \leq t}$ is $\eps$-{\deffont sparse} for $H$ if 
\[
\forall e \in E \qquad \Prob{P \sim \calP_e}{P \neq \inbraces{e}} ~\leq~ \eps \mper
\]
\end{definition}
In this section, we will prove that the hypergraphs arising from random CSP instances
admit sparse and consistent partitioning schemes. Recall that for a hypergraph $H$, we define
(\cref{def:metric-dmu})  the
metrics $d^H_{\mu}$ and $\rho^H_{\mu}$ as:
\[
    d^H_\mu(u,v) ~\coloneqq~ 1-(1-\mu)^{2\cdot d_H(u,v)}
     \qquad \text{ and } \qquad \rho^H_\mu(u,v) ~\coloneqq~
     \sqrt{\frac{2\cdot d^H_\mu(u,v)+ \mu}{1+\mu}}\mcom
\]
%
%
\begin{lemma}\label[lemma]{lem:partitioning}
Let $H = (V,E)$ be $k$-uniform hypergraph and let $d_{\mu}$ be the metric as defined
above. Let $H$ be such that for all sets $S \subseteq V$ with $\abs{S} \leq t$, the metric induced
on $\rho_{\mu}$ on $S$ is isometrically embeddable into $\ell_2$. Then, there exists 
$\eps \leq 10k \cdot \sqrt{\mu \cdot t}$ and $\Delta_H=O(1/\mu)$ such that  $H$ admits an
$\eps$-sparse consistent partitioning scheme of order $t$, with each partition consisting of clusters of diameter
at most $\Delta_H$ in $H$.
\end{lemma}
%
%

We use the following result of Charikar et al. \cite{CharikarCGGP98} which shows that
low-dimensional metrics have good \emph{separating decompositions} with bounded diameter \ie
decompositions which have a small probability of separating points at a small distance. 
\begin{theorem}[\cite{CharikarCGGP98}]\label[theorem]{thm:l2-decomposition}
Let $W$ be a finite
  collection of points in $\R^d$ and let $\Delta > 0$ be given. 
  Then there exists a distribution $\calP$ over partitions of $W$
  such that
\begin{itemize}
\item[-] $\forall P \in Supp(\calP)$, each cluster in $P$ has $\ell_2$ diameter at most $\Delta$.
\item[-]  For all $x,y \in W$
\[
\Prob{P \sim \calP}{P ~\text{separates}~ x ~\text{and}~ y} ~\leq~ 2\sqrt{d} \cdot \frac{\norm{x-y}_2}{\Delta} \mper
\]
\end{itemize}
\end{theorem}
We also need the observation that the partitions produced by the above theorem are consistent,
assuming the set $S$ considered above lie in a fixed bounded set (using a trivial modification of
the procedure in \cite{CharikarCGGP98}). For the sequel, we use $B(x,\delta)$ to denote the $\ell_2$
ball around $x$ of radius $\delta$ and $B_H(u,r)$ to denote a ball of radius $r$ 
around a vertex $u \in V(H)$. Thus,
\[
B(x,\delta) ~\defeq~ \inbraces{y ~\mid~ \norm{x-y}_2 \leq \delta}
\qquad \text{and} \qquad
B_H(u,r) ~\defeq~ \inbraces{v \in V ~\mid~ d_H(u,v) \leq r} \mper
\]
The balls $B(S,\delta)$ and $B_H(S,r)$ are defined similarly.
\begin{claim}\label[claim]{clm:l2-consistency}
Let $S$ and $T$ be sets such that $T \subseteq S$. Let $W_S = \inbraces{w_u}_{u \in S}$ and
$W_T = \inbraces{w_u'}_{u \in T}$ be $\ell_2$-embeddings of $S$ and $T$  
satisfying  $\phi(W_T) \subseteq W_S \subseteq  B(0,R_0) \subset
\R^d$, for some unitary transformation $\phi$ and $R_0 > 0$. 
Let $\calP_S$ and $\calP_T$ be distributions over partitions of $S$ and $T$ respectively, 
induced by partitions on $W_S$ and $W_T$ as given  by \cref{thm:l2-decomposition}. Then
\[
\calP_{S|T} ~=~ \calP_T \mper
\]
\end{claim}
\begin{proof}
The claim follows simply by considering (a trivial modification of) the algorithm of
\cite{CharikarCGGP98}. For a given set $W$ and a parameter $\Delta$, they produce a partition using
the following procedure:
\begin{itemize}
\item Let $W' = W$.
\item Repeat until $W' = \emptyset$
\begin{itemize}
\item Pick a random point $x$ in $B(W,\Delta/2)$ according to the Haar measure. Let $C_x =
  B(x,\Delta/2) \cap W'$. 
\item If $C_x \neq \emptyset$, set $W' = W' \setminus C_x$. Output $C_x$ as a cluster in the partition.
\end{itemize}
\end{itemize}
\cite{CharikarCGGP98} show that the above procedure produces a distribution over partitions
satisfying the conditions in \cref{thm:l2-decomposition}. We simply modify the procedure to sample a
random point $x$ in $B(0,R_0 + \Delta/2)$ instead of $B(S,\Delta/2)$. 
This does not affect the separation probability of any two points, since the only non-empty clusters
are still produced by the points in $B(S,\Delta/2)$. 
%

Let $P$ be a partition of $S$ produced by the above procedure when applied to the point set $W_S$,
and let $P'$ be a random partition produced when applied to the point set $\phi(W_T)$. 
It is easy to see from the above procedure that the distribution $\calP_T$ 
is invariant under a unitary transformation of $W_T$. 
By coupling the random choice of a
point in $B(0,R_0 + \Delta/2)$ chosen at each step in the procedures applied to $W_S$ and $\phi(W_T)
\subseteq W_S$, we get that $P(T) = P'$ \ie the partition $P$ restricted to
$T$ equals $P'$. 
Thus, we get $\calP_{S|T} = \calP_T$.
\end{proof}
We can use the above to prove \cref{lem:partitioning}.
\begin{proofof}{of \cref{lem:partitioning}}
Given a set $S$, let $W_S$ be an $\ell_2$ embedding of the metric $\rho_{\mu}$ restricted to
$S$. Since, $|S| \leq t$, we can assume $W_S \in \R^t$. We apply partitioning procedure of
Charikar \etal from \cref{thm:l2-decomposition} with $\Delta = 1/2$. From the definition of the metric
$\rho^H_{\mu}$, we get that there exists a $\Delta_H = O(1/\mu)$ such that $\rho^H_{u,v} \leq 1/2 
~\implies~ d_H(u,v) \leq \Delta_H$. Moreover, for $u,v$ contained in an edge $e$, we have that
$\rho_{\mu}(u,v) \leq \sqrt{5\mu}$ and hence the probability that $u$ and $v$ are separated is
at most $10\sqrt{\mu \cdot t}$. Thus, the probability that any vertex in $e$ is separated from $u$ is at
most $10k \cdot \sqrt{\mu \cdot t}$.

Finally, for any $S \subseteq T$, if $W_S$ and $W_T$ denote the corresponding $\ell_2$ embeddings,
by the rigidity of $\ell_2$ we have that for $\phi(W_T) \subseteq W_S$ for some unitary
transformation $\phi$. Thus, by \cref{clm:l2-consistency}, we get that this is a consistent
partitioning scheme of order $t$.
\end{proofof}
%
%


%% file: sagaps.tex

\subsection{Integrality Gaps from the Basic LP}
Recall that the basic LP relaxation for \maxkcspq($f$) as given in \cref{fig:basic-lp}.
%
%
In this section, we will prove \cref{thm:main}. We recall the statement below.
%
%
\maintheorem*
Let $\Phi_0$ be a $(c,s)$ integrality gap instance for the basic LP relaxation for
\maxkcspq($f$) with $n_0$ variables and $m_0$ constraints. We use it to construct a new integrality
gap instance $\Phi$. The construction is similar to the gap instances
constructed by Khot et al. \cite{KhotTW14} discussed in the next section. However, we describe this
construction first since it's simpler. The procedure for constructing the instance $\Phi$ is
described in \cref{fig:gap-instance}.
\begin{figure}[htb]
\hrule
\vline
\hspace{10 pt}
\begin{minipage}[t]{0.95\linewidth}
\vspace{10 pt}
{
\underline{\textsf{Given}}: A $(c,s)$ gap instance $\Phi_0$ on $n_0$ variables, for the basic LP.

\smallskip

\underline{\textsf{Output}}: An instance $\Phi$ with $N = n \cdot n_0$ variables and $m$ constraints. 

\medskip

The variables are divided into $n_0$ sets $X_1, \ldots, X_{n_0}$, one for each variable in
$\Phi_0$. We generate $m$ constraints  independently at random as follows:
\begin{enumerate}
\item Sample a random constraint $C_0 \sim \Phi_0$. Let $S_{C_0} = \inbraces{i_1, \ldots, i_k}
  \subseteq [n_0]$  denote the set of variables in this constraint.
\item For each $j \in [k]$, sample a random variable $x_{i_j} \in X_{i_j}$. 
\item Add the constraint $f((x_{i_1}, \ldots, x_{i_k}) + b_{C_0})$ to the instance $\Phi$.
\end{enumerate}
}
\smallskip
\end{minipage}
\hfill\vline
\hrule
\caption{Construction of the gap instance $\Phi$}
\label[figure]{fig:gap-instance}
\end{figure}

\subsubsection*{Soundness}
We first prove that no assignment satisfies more than $s+\eps$ fraction of constraints for the
above instance.
\begin{lemma}\label[lemma]{lem:basic-lp-soundness}
For every $\eps > 0$, there exists $\gamma = \gamma(\eps, n_0, q)$ such that for an instance $\Phi$
generated by choosing at least $\gamma \cdot n$ constraints independently at random as above, we
have with probability $1 - \exp\inparen{-\Omega(n)}$, $\opt(\Phi) < s+\eps$.
\end{lemma}
\begin{proof}
Fix an assignment $\sigma \in \qary^{N}$. We will first consider $\ex{\sat_{\Phi}\inparen{\sigma}}$ for a randomly
generated $\Phi$ as above.
\begin{align*}
\Ex{\Phi}{\sat_{\Phi}\inparen{\sigma}} 
&~=~ \expop_{C_0 \in {\Phi_0}} \expop_{x_{i_1} \in X_{i_1}} \cdots
                            \Ex{x_{i_k} \in X_{i_k}}{f(\sigma(x_{i_1})+b_{i_1}, \ldots,
                            \sigma(x_{i_k})+b_{i_k})} \\
&~=~ \expop_{C_0 \in {\Phi_0}} \Ex{Z_1,\ldots Z_{n_0}}{f(Z_{C_0} + b_{C_0})} \mcom
\end{align*}
where for each $i \in [n_0]$, $Z_i$ is an independent random variable with the distribution
\[
\prob{Z_i = b} ~\defeq~ \Ex{x \in X_i}{\indicator{\sigma(x) = b}} \mcom
\]
and $Z_{C_0}$ denotes the collection of variables in the constraint $C_0$ \ie $Z_{C_0} =
\inbraces{Z_i}_{i \in S_{C_0}}$.
Thus, the random variables $Z_1, \ldots, Z_{n_0}$ define a random assignment to the variables in
$\Phi_0$, which gives, for any $\sigma$
\[
\Ex{\Phi}{\sat_{\Phi}\inparen{\sigma}} ~=~ \expop_{C_0 \in {\Phi_0}} \Ex{Z_1,\ldots
  Z_{n_0}}{f(Z_{C_0} + b_{C_0})} ~<~ s \mper
\]
Consider a randomly added constraint $C$ to the instance $\Phi$. We have that 
\[
\prob{C(\sigma) = 1} ~=~ \Ex{\Phi}{\sat_{\Phi}(\sigma)} ~<~ s \mcom 
\]
for any fixed $\sigma$ over a random choice of the constraint $C$. Thus, for an instance $\Phi$ with
$m$ independently and randomly generated constraints, we have
\begin{align*} 
\Prob{\Phi}{\sat_{\Phi}(\sigma) ~\geq~ s + \eps}
&~\leq~
\Prob{\Phi}{\sat_{\Phi}(\sigma) 
~\geq~ \Ex{\Phi}{\sat_{\Phi}(\sigma)} + \eps} \\
&~=~ \Prob{\Phi}{\Ex{C \in \Phi}{\indicator{C(\sigma)=1}} ~\geq~ \Ex{\Phi}{\sat_{\Phi}(\sigma)} +
  \eps} \\
&~\leq~ \exp\inparen{-\Omega(\eps^2 \cdot m)} \mper
\end{align*}
Taking a union bound over all assignments, we get
\[
\Prob{\Phi}{\exists \sigma ~~\sat_{\Phi}(\sigma) ~\geq~ s + \eps}
~\leq~ q^{n \cdot n_0} \cdot \exp\inparen{-\eps^2 \cdot m} \mcom
\]
which is at most $\exp\inparen{-\Omega(n)}$ for $m = O(((\log q)/\eps^2)\cdot n \cdot n_0)$.
\end{proof}

\subsubsection*{Completeness}
To prove the completeness, we first observe that the instance $\Phi$ as constructed above is also a
gap instance for the basic LP. We will then ``boost'' this hardness to many levels of the
Sherali-Adams hierarchy.
\begin{lemma}\label[lemma]{lem:basic-lp-solution}
For every $\eps > 0$, there exists $\gamma = \gamma(\eps)$ such that for an instance $\Phi$
generated by choosing at least $\gamma \cdot n$ constraints independently at random as above, with probability $1 - \exp\inparen{-\Omega(n)}$ there exist distributions
$\treedist_{S_C}$ over $\qary^{S_C}$ for each $C \in \Phi$, and distributions $\treedist_{i}$ over $\qary$
for each variable $x_i \in [n \cdot n_0]$, satisfying
\begin{itemize}
\item[-] For all $C \in \Phi$ and all $i \in S_C$, $\treedist_{S_C | \{i\}} = \dist_i$.
\item[-] The distributions satisfy  $\ExpOp_{C \in \Phi} \Ex{\alpha \sim \treedist_{S_C}}{f(\alpha +
    b_C)} ~\geq~ c - \frac{\eps}{10}$.
\end{itemize}
\end{lemma}
\begin{proof}
For each $C_0 \in \Phi_0$ and each $j \in [n_0]$, let $\dzero_{S_{C_0}}$ and $\dzero_{j}$ denote the
basic LP solution satisfying
\[
\dzero_{S_{C_0} | j} ~=~ \dzero_{j} ~~\forall C_0 \in \Phi_0 ~\forall j \in S_{C_0}
\qquad
\text{and}
\qquad
\ExpOp_{C_0 \in \Phi_0} \Ex{\alpha \sim \dzero_{S_{C_0}}}{f(\alpha + b_{C_0})} ~\geq~ c \mper
\]
Each constraint $C \in \Phi$ is sampled according to some constraint $C_0 \in \Phi_0$, and we take
$\treedist_{S_C} \defeq \dzero_{S_{C_0}}$ for the corresponding contraint $C_0 \in \Phi_0$. Also,
each variable $x_i$ for $i \in [n_0 \cdot n]$, belongs to one of the sets $X_j$ for $j \in
[n_0]$, and we take $\treedist_{i} \defeq \dzero_{j}$ for the corresponding $j \in [n_0]$.

The consistency of the distributions follows immediately from the construction of the instance
$\Phi$. Let $C \in \Phi$ be any constraint and let $C_0$ be the corresponding constraint in
$\Phi_0$. If $S_{C_0} = (j_1, \ldots, j_k)$, then $S_C = (i_1, \ldots, i_k)$ where each $i_r \in
\{j_r\} \times [n]$ for all $r \in [k]$. Thus, for any $r \in [k]$,
\[
\treedist_{S_C | i_r} ~=~ \dzero_{S_{C_0} | j_r} ~=~ \dzero_{j_r} ~=~ \treedist_{i_r} \mper
\]
To bound the objective value, we again consider its expectation over a randomly generated instance
$\Phi$. Let $C$ be a random constraint added to $\Phi$. Then, if we define $\treedist_{S_C}$ as
above for this constraint, we have
\[
\ExpOp_{C} \Ex{\alpha \in \treedist_{S_C}}{f(\alpha + b_C)} ~=~ \ExpOp_{C_0 \in \Phi_0}\Ex{\alpha
  \sim \dzero}{f(\alpha + b_{C_0})} ~\geq~ c \mper
\]
Thus, the expected contribution of each constraint is at least $c$. The probability that the average
of $m$ constraints deviates by at least $\eps/10$ from the expectation, is at most
$\exp\inparen{- \Omega(\eps^2 \cdot m)}$. 
There exists $\gamma = O(1/\eps^2)$ such that for $m \geq
\gamma \cdot n$, the probability is at most $\exp(- \Omega(n))$.
\end{proof}
To construct local distributions for the Sherali-Adams hierarchy, we will consider (a slight
modification) the hypergraph $H$ corresponding to the instance $\Phi$. We first show that
distributions on edges of this hypergraph can be consistently
propagated in a tree, provided they agree on intersecting vertices. 

For a set $U \subseteq \Vtx(H)$ in a hypergraph
$H$, recall that $\cl(U)$ includes all paths of lengths at most 1 between any two vertices in
$U$. Thus, $E(\cl(U)) = \inbraces{e \in E ~\mid~ \abs{e \cap U} \geq 2}$.
Note that \cref{lem:basic-lp-solution} implies that edges forming a tree in $H$ satisfy the
hypothesis of \cref{lem:tree-propagation} below. 
\begin{lemma}\label[lemma]{lem:tree-propagation}
Let $H=(V,E)$ be a $k$-uniform hypergraph. Let $U \subseteq V$ and let the set of edges
$E(\cl(U))$ form a tree.  For each $e \in E(\cl(U))$, 
let $\treedist_e$ be a distribution on $[q]^e$ such that for any $u \in U$ and $e_1, e_2 \in E(\cl(U))$ such
that $e_1 \cap e_2 = \{u\}$, we have $\treedist_{e_1|u} = \treedist_{e_2|u} = \treedist_u$. 
Then, 
\begin{itemize}
\item[-] there exists a distribution $\treedist_U$ on $[q]^U$ such that $\treedist_{U | e \cap U} =
  \treedist_{e | e \cap U}$ for all $e \in E(U)$. 
\item[-] If $U' \subseteq U$ is such that the edges in $E(\cl(U'))$ form a subtree of $E(\cl(U))$, then 
$\treedist_{U | U'} = \treedist_{U'}$.
\end{itemize}
\end{lemma}
\begin{proof}
%
We define the distribution by starting with an arbitrary edge and traversing the tree in an
arbitrary order. Let $e_1, \ldots, e_r$ be a traversal of the edges in $E(\cl(U))$ such that for all $i$,
$\abs{\inparen{\cup_{j<i}e_j} \cap e_i} = 1$. Let $U_0 = \cup_{j<i}e_j$ be the set of vertices for
which we have already sampled an assignment and let $e_i$ be the next edge in the traversal, with $u$
being the unique vertex in $e_i \cap U_0$. We sample an assignment to the vertices in
$e$, conditioned on the value for the vertex $u$. Formally, we extend the distribution $\treedist_{U_0}$
to $U_0 \cup e$ by taking, for any $\alpha \in \qary^{U_0 \cup e}$
\[
\treedist_{U_0 \cup e} (\alpha)
~=~
\treedist_{U_0}(\alpha(U_0)) \cdot
\frac{\treedist_e(\alpha(e))}{\treedist_{e|u}(\alpha(u))} 
~=~
\treedist_{U_0}(\alpha(U_0)) \cdot
\frac{\treedist_e(\alpha(e))}{\treedist_{u}(\alpha(u))} 
\mper
\]
The above process defines a distribution $\treedist_{\cl(U)}$ on $\cl(U)$, with
\[
\treedist_{\cl(U)}(\alpha) ~=~ \frac{\prod_{e \in E(U)} \treedist_e(\alpha(e))}{ \prod_{u \in \cl(U)}
  \inparen{\treedist_{u}(\alpha(u))}^{\deg(u)-1}} \mper
\]
In the above expression, we use $\deg(u)$ to denote the degree of vertex $u$ in tree formed by the
edges in $E(\cl(U))$ \ie $\deg(u) = \abs{\inbraces{e \in E(\cl(U)) ~|~ u \in e}}$. We then define the
distribution $\treedist_U$ as the marginalized distribution $\treedist_{\cl(U) | U}$ \ie
\[
\treedist_U(\alpha) ~=~ \sum_{\beta \in \qary^{\cl(U)} \atop \beta(U) = \alpha}
\treedist_{\cl(U)}(\beta) \mper
\]
Note that the distribution $\treedist_{\cl(U)}$ and hence also the distribution $\treedist_U$ are independent
of the order in which we traverse the edges in $E(\cl(U))$. 
Also, since the above process samples each
edge according to the distribution $\treedist_e$, we have that for any $e \in E(U)$, $\treedist_{\cl(U) | e} =
\treedist_e$. Thus, also for any $e \in E(U)$, $\treedist_{U | e \cap U} = \treedist_{e | e \cap U}$.

Let $U' \subseteq U$ be any set such that $E(\cl(U'))$ forms a subtree of $E(\cl(U))$. Then there
exists a traversal $e_1, \ldots, e_r$, and $i \in [r]$ such that $e_j \in E(\cl(U')) ~\forall j \leq i$ and $e_j
\notin E(\cl(U')) ~\forall j > i$. However, the distribution defined by the partial traversal $e_1,
\ldots, e_i$ is precisely $\treedist_{\cl(U')}$. Thus, we get that $\treedist_{\cl(U) | \cl(U')} =
\treedist_{\cl(U')}$ which implies $\treedist_{U | U'} = \treedist_{U'}$.
\end{proof}
We can now prove the completeness for our construction using consistent decompositions.
\begin{lemma}\label[lemma]{lem:basic-lp-completeness}
Let $\eps > 0$ and let $\Phi$ be a random instance of \maxkcspq($f$) generated by choosing $\gamma
\cdot n$ constraints independently at random as above. Then, there is a $t =
\Omega_{\eps,k,n_0}\inparen{\frac{\log n}{\log \log n}}$, such that
with probability $1 - \eps$ over the choice of $\Phi$, there exist distributions
$\inbraces{\dist_S}_{\abs{S} \leq t}$ satisfying:
\begin{itemize}
\item[-] For all $S \subseteq V$ with $\abs{S} \leq t$, $\dist_S$ is a distribution on $\qary^S$.
\item[-] For all $T \subseteq S \subseteq V$ with $\abs{S} \leq t$, $\dist_{S | T} = \dist_T$.
\item[-] The distributions satisfy
\[
\ExpOp_{C \in \Phi} ~\Ex{\alpha_C \sim \dist_{S_C}}{f(\alpha_C + b_C)} ~\geq~ c - \eps \mper
\]
\end{itemize}
\end{lemma}
\begin{proof}
By \cref{thm:locally-l2}, we know that there exists $\delta$ such that with probability
$1-\eps/4$, after removing a set of constraints $C_B$ of size at most $(\eps/4)\cdot m$, we can assume
that the remaining instance has girth at least $\girth = \delta \cdot \log n$. Also, there exists
$\theta, c > 0$ such that for all $t \leq n^{\theta}$, the metric $\rho^H_{\mu}$ restricted to any
set $S$ of size at most $t$ embeds isometrically into the unit sphere in $\ell_2$, for all $\mu \geq
c \cdot \frac{\log t + \log \log n}{\log n}$.



We choose $\mu = 2c \cdot \frac{\log \log n}{\log n}$
and $t = \frac{\eps^2}{400 k^2} \cdot \frac{1}{\mu}$ so that
\[
\mu ~\geq~
c \cdot \frac{\log t + \log \log n}{\log n}
\quad \text{and} \quad
\sqrt{\mu \cdot t} ~\leq~ \frac{\eps}{20k} \mper
\]
Thus, by \cref{lem:partitioning}, $H$ admits an $(\eps/2)$-sparse partitioning scheme of order $t$ with
each cluster in the partition having diameter at most $\Delta_H = O(1/\mu)$. Let
$\inbraces{\calP_S}_{\abs{S} \leq t}$ denote this partitioning scheme.

Given a set $S$, the distribution $\dist_S$ is a convex combination of several distributions $\dist_{S,P}$,
corresponding to different partitions $P$ sampled from $\calP_S$. We describe the distribution
$\dist_S$ by giving the procedure to sample an $\alpha \in \qary^S$. Given the set $S$ with $\abs{S}
\leq t$:
\begin{itemize}
\item[-] Sample a partition $P = (U_1, \ldots, U_r)$ from the distribution $\calP_S$.
\item[-] For each set $U_i$, consider the set $\component{U_i}$ obtained by including the vertices
  contained in all the edges in the shortest path between all $u,v \in U_i$. Note that since $U_i$
  has diameter at most $\Delta_H$ in $H$, $\component{U_i}$ is connected and in fact $\component{U}
  = \cl_{\Delta_H}(U)$. 
Also, since the each vertex in an included path is within distance at most $\Delta_H/2$ of an
end-point, and $U_i$ has diameter at most $\Delta_H$, we know that the diameter of $\component{U_i}$
is at most $2 \cdot \Delta_H$. Hence, $\component{U_i}$ is a tree.
Finally, we must have $\cl(\component{U_i}) ~=~ \component{U_i}$ since any additional path of length
1 would create a cycle of length at most $2\cdot \Delta_H + 1$.
%
%
%

Thus, by \cref{lem:basic-lp-solution} and \cref{lem:tree-propagation} (with probability at least
$1-\eps/4$) there exists a 
distribution $\treedist_{\component{U_i}}$ for each $U_i$, satisfying $\treedist_{\component{U_i} |
  e} = \treedist_{e}$ for all $e \in E\inparen{\component{U_i}}$. Here, $\treedist_{e}$ are the
distributions given by \cref{lem:basic-lp-solution}, which form a solution to the basic LP for
$\Phi$, with value at least $c - \eps/4$. For each $U_i$, define the distribution
\[
\dist'_{U_i} ~\defeq~ \treedist_{\component{U_i} | U_i} \mper
\]
\item[-] Sample $\alpha \in \qary^S$ according to the distribution
\[
\dist_{S,P} ~\defeq~ \dist'_{U_1} \times \cdots \times \dist'_{U_r} \mper
\]
\end{itemize}
Thus, we have
\[
\dist_S ~:=~ \Ex{P = (U_1, \ldots, U_r) \sim \calP_S}{\prod_{i=1}^r \dist'_{U_i}} \mcom
\]
where the distributions $\dist'_{U_i}$ are defined as above. 

We first prove the distributions are consistent on intersections \ie $\dist_{S | T} = \dist_T$ for
any $T \subseteq S$. Note that by \cref{lem:partitioning}, the distributions $\calP_S$ and $\calP_T$
satisfy $\calP_{S|T} = \calP_T$. Each partition $(U_1, \ldots, U_r)$ also produces a partition
$T$. For ease of notation, we assume that the first (say) $r'$ clusters have non-empty intersection
with $S$. Let $V_i = U_i \cap T$ for $1 \leq i \leq r'$ ($V_i = \emptyset$ for $i > r'$). Then, we have
\begin{align*}
\dist_{S|T} 
~=~ 
\Ex{P = (U_1, \ldots, U_r) \sim \calP_S}{\prod_{i=1}^r \dist'_{U_i | V_i}} 
&~=~
\Ex{P = (U_1, \ldots, U_r) \sim \calP_S}{\prod_{i=1}^{r'} \treedist_{\component{U_i} | V_i}} \\
&~=~
\Ex{P = (U_1, \ldots, U_r) \sim \calP_S}{\prod_{i=1}^{r'} \treedist_{\component{V_i} | V_i}} \\
&~=~
\Ex{P' = (V_1, \ldots, V_{r'}) \sim \calP_T}{\prod_{i=1}^{r'} \treedist_{\component{V_i} | V_i}} 
\end{align*}
The second to last equality above uses the fact that $\component{V_i}$ is a subtree of
$\component{U_i}$ and thus $\treedist_{\component{U_i} | \component{V_i}} =
\treedist_{\component{V_i}}$ by \cref{lem:tree-propagation}. The last equality uses the fact that
$\calP_{S|T} = \calP_{T}$ by \cref{lem:partitioning}.

We now argue that the LP solution corresponding to the above distributions
$\inbraces{\dist_S}_{\abs{S} \leq t}$ has value at least $c - \eps$. Recall that the value of the LP
solution is given by
\[
\ExpOp_{C \in \Phi} ~\Ex{\alpha \sim \dist_{S_C}}{f(\alpha + b_C)} \mper
\]
Consider any constraint $C$ in $\Phi$, with the corresponding set of variables $S_C$ and the
corresponding hyperedge $e$. When defining the distribution $\dist_{S_C}$, we will partition $S_C$
according to the distribution $\calP_{S_C}$. By \cref{lem:partitioning} and our choice of parameters
\[
\Prob{P \sim \calP_{S_C}}{P \neq \{S_C\}} ~\leq~ 10k \cdot \sqrt{\mu \cdot t} ~\leq~ \frac{\eps}{2} \mper 
\]
For a constraint set which is not in the deleted set $C_B$, 
if the edge $e$ corresponding to the constraint $C$ is not split by a partition $P$ sampled according to
$\calP_{S_C}$, then by \cref{lem:tree-propagation} $\dist_{S_C, P} = \treedist_{S_C}$. Here,
$\treedist_{S_C}$ is the distribution given by \cref{lem:basic-lp-solution}. Since $f$ is Boolean,
we have that for $C \notin C_B$,
\[
\Ex{\alpha \sim \dist_{S_C}}{f(\alpha + b_C)} 
~\geq~
\Ex{\alpha \sim \treedist_{S_C}}{f(\alpha + b_C)} - \frac{\eps}{2} \mper
\]
Using \cref{lem:basic-lp-solution} again, we get
\begin{align*}
\ExpOp_{C \sim \Phi} ~\Ex{\alpha \sim \dist_{S_C}}{f(\alpha + b_C)} 
&~\geq~
\Ex{C \sim \Phi}{\inparen{1-\indicator{C \in C_B}} \cdot \inparen{\Ex{\alpha \sim \treedist_{S_C}}{f(\alpha + b_C)} - \frac{\eps}{2}}} \\
&~\geq~
\ExpOp_{C \sim \Phi} ~\Ex{\alpha \sim \treedist_{S_C}}{f(\alpha + b_C)} - \frac{\eps}{2} -  \Ex{C
  \sim \Phi}{\indicator{C \in C_B}} \\
&~\geq~
c - \frac{\eps}{4} - \frac{\eps}{2} - \frac{\eps}{4} \\
&~\geq~ 
c - \eps \mcom
\end{align*}
where the penultimate inequality uses the fact that the fraction of constraints in the initially deleted set
$C_B$ is at most $\eps/4$ (for large enough $n$).
\end{proof}

\subsection{Integrality Gaps for resistant predicates}
\input{ktw_predicates.tex}

\subsection{Lower bounds for LP extended formulations}
A connection between LP integrality gaps for the Sheral-Adams hierarchy, and lower bounds on the
size of LP extended formulations, was established by Chan \etal \cite{ChanLRS13}. They proved the
following:
\begin{theorem}[\cite{ChanLRS13}]
Let $k,q \in \N$ and $f: \qary^k \to \B$ be given.
Let $r: \N \to \N$ be a function such that the relaxation obtained by $r(n)$ levels of the
Sherali-Adams hierarchy cannot achieve a $(c,s)$ approximation for instances of \maxkcspq$(f)$ on $n$
variables. Then, for all large enough $n$, no LP extended formulation of size $n^{\inparen{r(n)}^2}$
can achieve a $(c,s)$ approximation on instances of size $N$, where $N \leq n^{10 \cdot r(n)}$
\end{theorem}
Combining the above with \cref{thm:main} and taking $r(n) = \Omega\inparen{\frac{\log n}{\log \log
    n}}$ yields \cref{cor:size-lower-bound}.
\sizecorollary*


%% file: ktw_predicates.tex
Let $f:\xspace \B^k\to \B$ be a boolean predicate and
let $\rho(f)=\frac{f^{-1}(1)}{2^k}$ be the fractions of satisfying assignments to $f$. Then
$f$ is approximation resistant if it is hard to distinguish the \maxcsp instances on $f$ between which are
at least $1-\littleoh(1)$ satisfiable vs which are at most $\rho(f)+\littleoh(1)$ satisfiable. 

In \cite{KhotTW14} the authors introduce the notion of {\sf vanishing measure} 
(on a polytope defined by $f$) and use it to characterize a variant of  approximation resistance,
called strong approximation resistance, assuming the Unique Games conjecture. 
They also show gave a \emph{weaker} notion of vanishing measures, which they used to characterize
strong approximation resistance for LP hierarchies. In particular, they proved that when the
condition in their characterization is satisfied, there exists a 
$(1-\littleoh(1),\rho(f)+\littleoh(1))$ integrality gap for $\bigoh(\log\log n)$ levels of
Sherali-Adams hierarchy for predicates $f$. 
Here, we show that using \cref{thm:main}, their result can be simplified and strengthened
\footnote{The LP integrality gap result of Khot \etal is in fact slightly stronger than stated
  above. They show that LP value is at least $1-o(1)$ while there is no integer solution achieving a
value outside the range $[\rho(f)-o(1), \rho(f)+o(1)]$. It is easy to see that the same also holds
for the instance constructed here.} 
to $O\inparen{\frac{\log n}{\log \log n}}$ levels.

Let us first recall some useful notation defined by Khot \etal \cite{KhotTW14} before we define the
notion of vanishing measure:
\begin{definition}
  \label[definition]{def:vanishing-measure-lp-1}
  For a predicate $f:\xspace \B^k\to \B$, let $\calC(f)$ be the convex polytope
  of \emph{first} moments (biases) of distributions supported on satisfying assignments of
  $f$ \ie
\[
\calC(f) ~\defeq~ \inbraces{ \zeta \in \R^k ~\mid~  \forall i \in [k], ~\zeta_i = \Ex{\alpha \sim
    \nu}{(-1)^{\alpha_i}}, ~~\supp(\nu) \subseteq f^{-1}(1) 
} \mper
\]
For a measure $\Lambda$ on $\calC(f)$, $S \subseteq [k]$, $b \in \B^S$ and permutation $\pi: S
\to S$, let $\Lambda_{S,\pi,b}$ denote the induced measure on $\R^S$ by considering vectors
with coordinates $\inbraces{(-1)^{b_{\pi(i)}} \cdot \zeta_{\pi(i)}}_{i \in S}$, where $\zeta \sim
\Lambda$.
\end{definition}
We recall below the definition of vanishing measure for LPs from \cite{KhotTW14} (see
Definition $1.3$) :
\begin{definition}
  \label[definition]{def:vanishing-measure-lp-2}
  A measure $\Lambda$ on $\calC(f)$ is called {\deffont vanishing} (for LPs) if for
  every $1\le t\le k$, the following signed measure
  \[
    \ExpOp_{\abs{S}=t} ~\ExpOp_{\pi:\xspace S\to S} ~\Ex{b\in \B^t}{\inparen{\prod_{i=1}^t (-1)^{b_i}}\cdot
    \hat{f}(S)\cdot \Lambda_{S,\pi,b}}
  \]
  is identically $0$. We say $f$ has a vanishing measure if there exists a vanishing measure $\Lambda$ on $\calC(f)$.
\end{definition}
In particular, they prove the following theorem:
\begin{theorem}
  \label[theorem]{thm:ktw-SA-gap}
  Let $f : \B^k \to \B$ be a $k$-ary boolean predicate that has a vanishing measure. Then for every $\eps>0$,
  there is a constant $c_\eps > 0$ such that for infinitely may $N \in\N$, there exists an
  instance $\Phi$ of $\maxkcsp(f)$ on $N$ variables satisfying the following:
  \begin{itemize}
    \item $\opt(\Phi) ~\leq~ \rho(f) + \eps$.
    \item The optimum for the LP relaxation given by $c_\eps\cdot \log\log N$ levels of Sherali-Adams hierarchy
      has $\sdpopt(\Phi) \geq 1-\bigoh(k\cdot \sqrt \eps)$.
  \end{itemize}
\end{theorem}
Combining this with our \cref{thm:main} already gives us the following stronger result:
\begin{corollary}
  \label[corollary]{cor:ktw-SA-gap-improved}
  Let $f : \B^k \to \B$ be a $k$-ary boolean predicate that has a vanishing measure. Then for every $\eps>0$,
  there is a constant $c_\eps > 0$ such that for infinitely may $N \in\N$, there exists an
  instance $\Phi$ of $\maxkcsp(f)$ on $N$ variables satisfying the following:
  \begin{itemize}
    \item All integral assignment of $\Phi$ satisfies at most $\rho(f) +  \eps$ fraction of
      constraints.
    \item The LP relaxation given by  $c_\eps\cdot \frac{\log N}{\log\log N}$ levels of
      Sherali-Adams hierarchy has $\sdpopt(\Phi) \geq 1- O(k\sqrt{\eps})$.
  \end{itemize}
\end{corollary}
%
%

However, note that to apply \cref{thm:main}, one only needs a gap for the basic LP, which is much
weaker requirement than the $O(\log \log N)$-level gap given by \cref{thm:ktw-SA-gap}. 
We observe below that the gap for the basic LP follows very simply from the construction by
Khot \etal \cite{KhotTW14}. One can then directly use this gap for applying \cref{thm:main} instead
of going through \cref{thm:ktw-SA-gap}.

Khot \etal \cite{KhotTW14} use the probabilistic construction given in \cref{fig:ktw-instance}, 
for a given $\eps > 0$.
The construction actually requires $\Lambda$ to be a vanishing measure over the polytope
$\calC_{\delta}(f) \defeq (1-\delta) \cdot \calC(f)$, for $\delta = \sqrt{\eps}$. However, since
$\calC_{\delta}(f)$ is simply a scaling of $\calC(f)$, a vanishing measure over $\calC(f)$
also gives a vanishing measure over $\calC_{\delta}(f)$. Note that each $\zeta_0 \in \calC(f)$
corresponds to a distribution $\nu_0$ supported in $f^{-1}(1)$. For each $\zeta \in \calC_{\delta}$,
let $\zeta_0 = \frac{1}{1-\delta} \cdot \zeta$ be the point in $\calC(f)$ with distribution
$\nu_0$. Then the distribution $\nu = (1-\delta) \cdot \nu_0 + \delta \cdot U_k$ (where $U_k$
denotes the uniform distribution on $\B^k$) satisfies $\forall i \in [k] \Ex{\alpha \sim
  \nu}{(-1)^{\alpha_i}} = \zeta_i$.
\begin{figure}[!ht]
  \hrule
  \vline
  \hspace{10 pt}
  \begin{minipage}[t]{0.95\linewidth}{
    \vspace{10 pt}
    Let $n_0=\lceil \frac{1}{\eps}\rceil$. Partition the interval
    $[0,1]$ into $n_0+1$ disjoint intervals $I_0,I_1,\ldots ,I_{n_0}$ where $I_0=\{0\}$ and $I_i=\left(
    \nfrac{i-1}{n_0},\nfrac{i}{n_0}\right]$ for $1\le i\le n_0$. For each interval $I_i$, let $X_i$
  be a collection of $n$ variables (disjoint from all $X_j$ for $j \neq i$).

\smallskip
Generate $m$ constraints independently according to the following procedure:
    \begin{itemize}
      \item Sample $\zeta\sim \Lambda$.
      \item For each $j\in [k]$, let $i_j$ be the index of the interval which contains
        $\abs{\zeta(j)}$. Sample uniformly a variable $y_j$ from the set $X_j$.
      \item If $\zeta(j)<0$, then negate $y_j$. If $\zeta(j)=0$, then negate $y_j$ w.p.
        $\frac{1}{2}$.
      \item Introduce the constraint $f$ on the sampled $k$ tuple of literals.
    \end{itemize}
    \smallskip
  }
  \end{minipage}
  \hfill
  \vline
  \hrule
  \caption{Sherali-Adams integrality gap instance for vanishing measure}
  \label[figure]{fig:ktw-instance}
\end{figure}

They show for a sufficiently large constant $\gamma$, an instance $\Phi$ with $m = \gamma \cdot n$
constraints satisfies with high probability, that for all assignments $\sigma$,
$\abs{\sat_{\Phi}(\sigma) - \rho(f)} \leq \eps$ (see Lemma $4.4$ in \cite{KhotTW14}). The proof is
similar to that of of \cref{lem:basic-lp-soundness}.

Additionally, we need the following claim from \cite{KhotTW14} (see Claim $4.7$ there), which allows
one to ``round''
coordinates of the vectors $\zeta \in \calC_{\delta}(f)$ to the end-points of the intervals $I_0,
\ldots, I_{n_0}$. This ensures that any two variables in the same collection $X_i$ have the same
bias.  The proof of the claim follows simply from a hybrid argument. We include it in the appendix
for completeness.
\begin{restatable}{repclaim}{biasclaim}
\label{claim:bias-rounding}
Let $\zeta \in \calC_{\delta}(f)$ and let $\nu$ be the corresponding distribution supported in $f^{-1}(1)$
such that for all $i \in [k]$, we have $\zeta_i = \Ex{\alpha \sim \nu}{(-1)^{\alpha_i}}$. Let $t_1, \ldots,
t_k \in [0,1]$ be such that for all $i \in [k]$, $\abs{t_i - \abs{\zeta_i}} \leq \eps$ for
$\eps<\delta/2$. 
Then there exists a distribution $\nu'$ on $\B^k$ such that
\[
\norm{\nu - \nu'}_1 ~=~ O(k \cdot (\eps/\delta))
\qquad \text{and} \qquad
\forall i \in [k], ~~\Ex{\alpha \sim \nu'}{(-1)^{\alpha_i}} ~=~ \sgn(\zeta_i)\cdot t_{i} \mper
\]
%
\end{restatable}
%
%
%
We can now use the above to give a simplified proof of \cref{cor:ktw-SA-gap-improved}.
\begin{proofof}{of \cref{cor:ktw-SA-gap-improved}}
  Here we exhibit a solution of the basic LP \cref{fig:basic-lp} for the instance given in
  \cref{fig:ktw-instance}.
  For each variable $y_{j}$ coming from the set $X_j$ for $j\in\{0,1,\ldots,n_0\}$, we set
  the bias $t_{j}$ of the variable to be the rightmost point of the interval $I_j$ \ie 
  set $\vartwo{y_j}{-1}= \frac12 \cdot \inparen{1-\frac{i}{n_0}}$ and  $\vartwo{y_j}{1}=\frac12 \cdot
  \inparen{1+\frac{i}{n_0}}$. 
  
  For each constraint $C$ of the form $f(y_{i_1} + b_{1}, \ldots, y_{i_k}+b_{k})$, 
  let $\zeta(C) \in \calC_{\delta}(f)$ be the point used to generate it,
  and let $\nu(C)$ denote the corresponding distribution on $\B^k$. By \cref{claim:bias-rounding},
  there exists a distribution $\nu'(C)$ such that $\norm{\nu(C) - \nu'(C)}_1 = O(k\eps/\delta)$ and
  such that the biases of the \emph{literals} satisfy $\Ex{\alpha \sim \nu'(C)}{(-1)^{\alpha_j}} =
  \sgn(\zeta_j) \cdot t_{i_j}$, where $t_{i_j}$ denotes the bias for the interval to which $y_{i_j}$
  belongs. When $t_{i_j} \neq 0$, we negate a variable only when $\sgn(\zeta_j) < 0$. Thus, we have
  $\Ex{\alpha \sim \nu'(C)}{(-1)^{\alpha_j + b_j}} = t_{i_j}$, which is consistent with the bias given by the
  singleton variables $\vartwo{y_{i_j}}{1}$ and $\vartwo{y_{i_j}}{-1}$. We thus define the local
  distribution on the set $S_C$ as $\dist_{S_C}(\alpha) = (\nu'(C))(\alpha+b_C)$.

  For all $C \in \Phi$, since $\zeta(C) \in \calC_{\delta}(f)$, we have that $\Ex{\alpha \sim
    \nu(C)}{f(\alpha)} \geq 1 - \delta$. Also, since $\norm{\nu(C)-\nu'(C)}_1 =
  O(k\eps/\delta)$, we get that $\Ex{\alpha \sim \nu'(C)}{f(\alpha)} \geq 1 - \delta -
  O(k\eps/\delta)$. Thus, we have for all $C \in \Phi$, $\Ex{\alpha \sim \dist_{S_C}}{f(\alpha+b_C)} \geq 1
  - \delta - O(k\eps/\delta)$. Taking $\delta = \sqrt{\eps}$ proves the claim.
\end{proofof}

%% file: appendix.tex
\section{Omitted proofs from \cref{sec:hypergraphs} and
  \cref{sec:sagaps}}\label[appendix]{sec:appendix}

\locallysparse*
\begin{proof}
  As in the proof of \cref{lem:cycle-count}, given a random hypergraph $H$,
  we construct a hypergraph $H'$ ( by contracting all the vertices in
  $[n_0]\times\{j\}$ to $j\in [n]$ ).

  Consider a subset of vertices $S \subseteq \Vtx(H)$ and let $S' \subseteq \Vtx(H')$ be the
  corresponding contracted set in $H'$. Since each edge in $H$ corresponds to an edge in $H'$
  (counting multiplicities), we have
\[
\abs{E(S)} ~\geq~ \frac{\abs{S}}{k-1-\eta} 
\quad  \Rightarrow \quad 
\abs{E(S')} ~\geq~ \frac{\abs{S}}{k-1-\eta} ~\geq~ \frac{\abs{S'}}{k-1-\eta} \mper
\]
Thus, it suffices to show that $H'$ is $(\tau',\eta)$-sparse for $\tau' = \tau \cdot n_0$, since
$\abs{S'} \leq \tau \cdot N = (\tau \cdot n_0) \cdot n$.
Given any multiset in $[n]^k$, the probability that it corresponds to an edge in $H'$ is at most 
$(k!) \cdot (m/n^k) $. Thus, the probability that there exists a set $T$ of size at most $\tau'
\cdot n$, containing at least $\abs{T}/(k-1-\eta)$ edges (counting multiplicities) is at most
\[
\sum_{h=1}^{\tau' \cdot n} \binom{n}{h} \cdot \binom{h^k}{r} \cdot \inparen{\frac{k! \cdot
    m}{n^k}}^r \mcom
\]
where $r = \frac{h}{k-1-\eta}$. Note that  we also need to consider $h=1$ as edges in $H'$
correspond to multisets of size $k$, and so may not have all distinct vertices. Simplifying the
above using $\binom{a}{b} \leq \inparen{\frac{a \cdot e}{b}}^b$ and $k! \leq k^k$ gives
\begin{align*}
\sum_{h=1}^{\tau' \cdot n} \binom{n}{h} \cdot \binom{h^k}{r} \cdot \inparen{\frac{k! \cdot
    m}{n^k}}^r
&~\leq~
\sum_{h=1}^{\tau' \cdot n} \inparen{\frac{n \cdot e}{h}}^h \cdot \inparen{\frac{h^k \cdot e}{r}}^r \cdot \inparen{\frac{k^k \cdot
    m}{n^k}}^r \\
&~=~
\sum_{h=1}^{\tau' \cdot n} \inparen{e^{k-\eta} \cdot (k-1-\eta) \cdot k^k \cdot \gamma \cdot
  \inparen{\frac{h}{n}}^{\eta}}^{h/(k-1-\eta)} \\
&~\leq~
\sum_{h=1}^{\tau' \cdot n} \inparen{k^{3k} \cdot \gamma \cdot
  \inparen{\frac{h}{n}}^{\eta}}^{h/(k-1-\eta)} \\
\end{align*}
Let $\theta = \eta/(2k)$. We divide the above summation in two parts and first consider
\begin{align*}
\sum_{h=n^{\theta}}^{\tau' \cdot n} \inparen{k^{3k} \cdot \gamma \cdot
  \inparen{\frac{h}{n}}^{\eta}}^{h/(k-1-\eta)}
&~\leq~
\sum_{h=n^{\theta}}^{\tau' \cdot n} \inparen{k^{3k} \cdot \gamma \cdot
  \inparen{\tau'}^{\eta}}^{n^{\theta}/(k-1-\eta)} \\
&~\leq~
2 \cdot \exp\inparen{- \frac{n^{\theta}}{k}} \\
&~\leq~
2 \cdot \frac{k}{n^{\theta}}
\mcom
\end{align*}
for $\tau' \leq \inparen{e \cdot k^{3k} \cdot \gamma}^{-1/\eta}$.
Considering the first half of the summation, we get
\begin{align*}
\sum_{h=1}^{n^{\theta}} \inparen{k^{3k} \cdot \gamma \cdot
  \inparen{\frac{h}{n}}^{\eta}}^{h/(k-1-\eta)}
&~\leq~
n^{\theta} \cdot \inparen{\frac{k^{3k} \cdot \gamma}{n^{(1-\theta) \cdot \eta}}}^{1/k} \\
&~\leq~
\inparen{\frac{k^{3k} \cdot \gamma}{n^{\eta/4}}}^{1/k} 
~=~
k^3 \cdot \gamma^{1/k} \cdot n^{-\theta/2} \mper
\end{align*}
Combining the two bounds gives that the probability is at most
~$3k^3 \cdot \gamma^{1/k} \cdot n^{-\theta/2}$, which equals the desired bound.
\end{proof}

\biasclaim*

\begin{proof}
Let $r_j = \sgn(\zeta_j) \cdot t_{j}$ be the desired bias of the $j^{th}$ coordinate. Then,
$\abs{\zeta(j) - r_j} \leq \eps$ for all $j \in [k]$ 
We construct a sequence of distributions $\nu_0,\ldots,\nu_k$ such that $\nu_0 = \nu$ and
$\nu_k = \nu'$. In $\bnu_j$, the biases are $(r_1,\ldots,r_j,\zeta_{j+1},\ldots,\zeta_k)$.

The biases in $\nu_0$ satisfy the above by definition. 
We obtain $\bnu_{j}$ from $\bnu_{j-1}$ as, 
\[
\nu_{j} = (1-\tau_j) \cdot \nu_{j-1} + \tau_j \cdot D_j \mcom
\]
where $D_j$ is the distribution in which all bits, except for the $j^{th}$ one, are set
independently according to their biases in $\bnu_{j-1}$. For the $j^{th}$ bit, we set it to
$\sgn(r_j-\zeta_j)$ (if $r_j-\zeta(j) = 0$, we can simply proceed with $\bnu_j = \bnu_{j-1}$). 
The biases for all except for the $j^{th}$ bit are unchanged. For the $j^{th}$ bit, the bias now
becomes $r_j$ if
\[
r_j = (1-\tau_j) \cdot \zeta_j + \tau_j \cdot \sgn(r_j-\zeta_j)
~\Longrightarrow~
\tau_j \cdot (\sgn(r_j-\zeta_j) - r_j) = (1-\tau_j) \cdot (r_j - \zeta_j) \mper
\]
Since $\zeta \in \calC_{\delta}(f)$, we know that 
$\abs{\sgn(r_j-\zeta(j)) - r_j} \geq \delta/2$. Also, $\abs{r_j - \zeta(j))} \leq \eps$ by
assumption. Thus, we can choose $\tau_j = O(\eps/\delta)$ which gives that $\norm{\bnu_{j} -
  \bnu_{j-1}}_1 = O(\eps/\delta)$. The final bound then follows by triangle inequality. 
\end{proof}
